\documentclass[runningheads]{llncs}

\usepackage[utf8]{inputenc}
\usepackage[table]{xcolor} %For gray color in Table: https://tex.stackexchange.com/questions/8891/color-merged-and-regular-cells-in-a-table-individually
\usepackage{amsmath}
\usepackage{amssymb}
\usepackage{adjustbox}
\usepackage{pdfpages}
\usepackage{caption}
\usepackage[ruled,vlined,linesnumbered]{algorithm2e}

\newcommand{\rememberlines}{\xdef\rememberedlines{\number\value{AlgoLine}}}
\newcommand{\resumenumbering}{\setcounter{AlgoLine}{\rememberedlines}}

\usepackage{upgreek}
\usepackage{latexsym }
\usepackage{ stmaryrd }
\newcommand{\comment}[1]{}
\usepackage{centernot,cancel}
\usepackage[toc,page]{appendix}
\usepackage{tabto}
\usepackage{comment}
\usepackage{multirow} %FOR CELL-MERGING IN TABLES
\usepackage{paralist}
\usepackage{todonotes}

\usepackage{wrapfig} %Wrap figures within text: https://tex.stackexchange.com/questions/27291/problems-with-captions-in-the-environment-wrapfigure

\author{}
\authorrunning{}
\institute{}

\newcommand{\Thanks}{\thanks{This work was funded in part by the Madrid Regional Gov.
  Project ``S2018/TCS-4339 (BLOQUES-CM)'', by PRODIGY Project
  (TED2021-132464B-I00) funded by MCIN/AEI/10.13039/501100011033/ and
  the European Union Next Generation EU/PRTR, and by a research grant
  from Nomadic Labs and the Tezos Foundation.}}

\title{Boolean Abstractions \\ for Realizability Modulo Theories \\ (Extended version) \Thanks}

\usepackage{orcidlink}

\author{Andoni Rodríguez \inst{1,2} \orcidlink{0009-0006-3464-8667} \and César Sánchez \inst{1} \orcidlink{0000-0003-3927-4773}}
\institute{
     IMDEA Software Institute, Madrid. Spain
     \and
     Universidad Politécnica de Madrid. Spain
}

\usepackage{graphicx}
\usepackage{amssymb} %% <- for \square and \blacksuare
\usepackage{ltlfonts}
\usepackage{comment}

\usepackage{hyperref}
 \hypersetup{%
   colorlinks=true,           % use colored links
   allcolors=blue!70!black,   % use dark blue for all links
   pdfstartview=Fit,          % open PDF viewer with side-pane hidden.
   breaklinks=true,
   pdfauthor={},
   pdftitle={Boolean Abstractions for Realizability Modulo Theories}}

\newcommand{\KWD}[1]{\textit{#1}}

\newcommand{\True}{\KWD{true}}
\newcommand{\False}{\KWD{false}}

\newcommand{\DefinedAs}{\,\stackrel{\text{def}}{=}\,}

\newcommand{\AP}{\ensuremath{\mathsf{AP}}\xspace}

\newcommand{\DefOR}{\ensuremath{\hspace{0.2em}\big|\hspace{0.2em}}}

\newcommand{\Always}{\LTLsquare}
\newcommand{\Event}{\LTLdiamond} 
\newcommand{\Next}{\LTLcircle}

\newcommand{\U}{\mathbin{\mathcal{U}}}

\renewcommand{\And}{\mathrel{\wedge}}
\newcommand{\Or}{\mathrel{\vee}}
\newcommand{\Impl}{\mathrel{\rightarrow}}
\newcommand{\Into}{\Impl}

\newcommand{\LTLt}{\ensuremath{\textup{LTL}_{\calT}}\xspace}
\newcommand{\LTLB}{\ensuremath{\textup{LTL}_{\mathbb{B}}}\xspace}

\newcommand{\configuration}{choice\xspace}
\newcommand{\configurations}{choices\xspace}

\newcommand{\reactions}{reactions\xspace}
\newcommand{\xbar}{\overline{x}}

\newcommand{\Lit}{\ensuremath{\mathit{Lit}}\xspace}
\newcommand{\phiT}{\ensuremath{\varphi_{\mathcal{T}}}\xspace}
\newcommand{\phiB}{\ensuremath{\varphi_{\mathbb{B}}}\xspace}
\newcommand{\phiExtra}{\varphi^{\textit{extra}}}
\newcommand{\phiEx}{\ensuremath{\phiExtra}\xspace}
\newcommand{\xs}{\ensuremath{\overline{x}}\xspace}
\newcommand{\ys}{\ensuremath{\overline{y}}\xspace}

\newcommand{\PT}[1]{#1^p}
\newcommand{\NT}[1]{#1^a}

\newcommand{\mycal}[1]{\ensuremath{\mathcal{#1}}\xspace}
\newcommand{\calC}{\mycal{C}}
\newcommand{\react}{\ensuremath{\textit{react}}}
\newcommand{\calR}{\mycal{R}}
\newcommand{\calT}{\mycal{T}}
\newcommand{\calQ}{\mycal{Q}}
\newcommand{\VR}{\ensuremath{\textit{VR}}\xspace}

\newcommand{\VE}{\ensuremath{\overline{v}_e}}
\newcommand{\VS}{\ensuremath{\overline{v}_s}}

\newcommand{\ThZ}{\mathcal{T}_\mathbb{Z}}
\newcommand{\ThR}{\mathcal{T}_\mathbb{R}}

\newcommand{\GameT}{\mathcal{G}^{\calT}}
\newcommand{\GameB}{\mathcal{G}^{\mathbb{B}}}

\newcommand{\qrprec}{\mathrel{\preceq}}

\newcommand{\qreact}{\ensuremath{\textit{qreact}}}
\newcommand{\Vars}{\mathit{Vars}}
\newcommand{\Bool}{\mathbb{B}}
\newcommand{\Part}{\rightharpoonup}

\definecolor{darkGray}{gray}{0.55} %https://tex.stackexchange.com/questions/94799/how-do-i-color-table-columns
\definecolor{lightGray}{gray}{0.85} %https://tex.stackexchange.com/questions/94799/how-do-i-color-table-columns

\newcommand{\rhoT}{\ensuremath{\rho_{\mathcal{T}}}\xspace}
\newcommand{\rhoB}{\ensuremath{\rho_{\mathbb{B}}}\xspace}
\newcommand{\innerLoop}{\ensuremath{\textit{inner}\_\textit{loop}}}

\newcommand{\AlgBruteForce}{
\SetAlgorithmName{Alg.}{} 
\textbf{Input: }$\phiT$ \\
$\varphi'\gets \phiT[l_i \leftarrow s_i]$  \textit{VR} $\gets$ \{\} \\
$\calC \gets \textit{choices}(\textit{literals}(\phiT))$ \\
$\calR \gets 2^{\calC}$\ \\
 \For{$(P,A) \in \calR$}{
   \If{$\exists\xs.\react_{(P,A)}(\xs)$}
   {
     $\textit{VR} \gets \textit{VR} \cup \{ (P,A) \}$ \\
    }
 }
 $\phiExtra \gets \textit{getExtra}(\!\textit{VR})$ \\
 \Return $\varphi' \wedge \square (A \Impl \phiExtra)$
 \rememberlines
}

\newcommand{\AlgModelLoop}{
  \resumenumbering
  \SetAlgorithmName{Alg.}{}
\textbf{Input: }$\phiT$ \\
$\varphi'\gets \phiT[l_i \leftarrow s_i]$ ; \textit{VR} $\gets$ \{\} \\
$\calC \gets \textit{choices}(\textit{literals}(\phiT))$ \\
$\calR \gets 2^{\calC}$ ;
$\psi \gets \top $ \\
 \While{$\textit{SAT}(\psi)$}{
    $m = model(\psi)$ \\
    \If{$\exists \overline{x} \textit{. } (\textit{toTheory}(m,\calC))$}{
      $P \gets \textit{posVars}(m)$ \\
      $\psi \gets \psi \wedge \neg(\bigwedge_{p \in P} p)$ \\
      \textit{VR} $\gets$ \textit{VR} $\cup$ $(e_{\textit{t}}, P)$ \\
    }
    \Else{
      $N \gets \textit{negVars}(m)$ \\
      $\textit{fh} \gets \bigwedge_{n \in N} n$ \\
    \If{$\exists \overline{x} \textit{. } \textit{toTheory}(\textit{fh},\calC)$}{
      $\psi \gets \psi\wedge \neg m$
    }
    \Else{
      $\psi \gets  \psi\wedge \neg \textit{fh}$ \\
    }
    }
 }
 $\phiExtra \gets \textit{getExtra}\textit{(VR)}$ \\
  \Return $\varphi' \wedge \square (A \Impl \phiExtra)$
 \rememberlines
}

\newcommand{\AlgNestedSAT}{
  \resumenumbering
\SetAlgorithmName{Alg.}{} 
\textbf{Input: }$\phiT$ \\
$\varphi'\gets \phiT[l_i \leftarrow s_i]$ ; \textit{VR} $\gets$ \{\} \\
$\calC \gets \textit{choices}(\textit{literals}(\phiT))$ \\
$\calR \gets 2^{\calC}$ ;
$\psi \gets \top $ \\
 \While{$\textit{SAT}(\psi)$}{
    $m = model(\psi)$ \\
    \If{$\exists \overline{x} \textit{. } (\textit{toTheory}(m, \calC))$}{
      $P \gets \textit{posVars}(m)$ \\
      $\psi \gets \psi \wedge \neg(\bigwedge_{p \in P} P)$ \\
      \textit{VR} $\gets$ \textit{VR} $\cup$ $(e_{\textit{t}}, P)$ \\
    }
    \Else{
      $N \gets \textit{negVars}(m)$ \\
      $\psi \gets \psi\wedge \neg m$ \\
      $I \gets \innerLoop (m,\calC)$ \\ 
      $\psi \gets \psi \wedge \neg(\bigwedge_{i \in I} i)$ \\
    }
 }
  $\phiExtra \gets \textit{getExtra}\textit{(VR)}$ \\
   \Return $\varphi' \wedge \square (A \Impl \phiExtra)$
  \rememberlines
}

\newcommand{\AlgInnerLoop}{
  \resumenumbering
\SetAlgorithmName{Alg.}{} 
\textbf{Input: }$m, \calC$ \\
\textit{VQ} $\gets$ \{\} ; $\beta \gets \top $ \\
 \While{$\textit{SAT}(\beta)$}{
    $u = model(\beta)$ \\
    \If{$\exists \overline{x} \textit{. } (\textit{toTheory\_inn}(u,m,\calC))$}{
      $P \gets \textit{posVars}(u)$ \\
      $\beta \gets \beta \wedge \neg(\bigwedge_{p \in P} p)$ \\
    }
    \Else{
      $N \gets \textit{negVars}(u)$ \\
      $\beta \gets \beta \wedge \neg(\bigwedge_{n \in N} n)$ \\
      \textit{VQ} $\gets$ \textit{VQ} $\cup$ $u$ \\
    }
 }
 \Return \textit{VQ}
 \rememberlines
}

\begin{document}

\maketitle

\newcommand{\TableBenchmark}{
\begin{figure}[t!]
%  \centering
\begin{tabular}{|c|c||c|c|c||c|c||c|c|c|c||c|c|}  \hline
  Bn. & Cls. & \multicolumn{3}{|c||}{Time (s) } & \multicolumn{2}{|c||}{Queries (out+inn)} & \multicolumn{4}{|c|}{Heuristics (doub)} & \multicolumn{2}{|c|}{$\varphi^{\mathbb{B}}$} \\ 
 %\multirow{2}{*}{(name)} & \multirow{2}{*}{(vars / lits)} & \multicolumn{3}{c}{(s)} & \multicolumn{3}{c}{(outer / inner)} & \multicolumn{4}{c}{(of double SAT)}\\
 \cline{3-13}
  (nm.) & (vr, lt) & BF & SAT & Doub. & SAT & Doub. & MxI. & Md. & Dc. & $A$. & Val. & Tme. \\ [0.5ex] 
 \hline\hline
\cellcolor{darkGray}  & (1, 7) & $\perp$ & 6740 & \textbf{31.77} & 30375 & \textbf{72/1040} & 40 & 2 & 0 & $\checkmark$ & 1 & \cellcolor{darkGray} \\ %4.22 \\ 

  \cellcolor{darkGray} &  (2, 4) & 3911 & \textbf{0.70} & 0.91 & \textbf{27} & 25/20 & 10 & 2 & 0 & $\times$ & 16 & \cellcolor{darkGray} \\ %4.41 \\
  \cellcolor{darkGray} & (1, 3) & 3.64 & 1.19 & \textbf{0.52} & 46 & \textbf{10/20} & 10 & 2 & 0 & $\times$ & 4 & \cellcolor{darkGray} \\ %3.16 \\
  \multirow{-4}{*}{\cellcolor{darkGray} {\textit{Li.}}} & (1, 2) & 0.23 & \textbf{0.09} & 0.14 & 4 & \textbf{4/3} & 3 & 3 & 0 & $\times$ & 3 & \multirow{-4}{*}{\cellcolor{darkGray} {4.41}}  \\ %3.78 \\
 \hline \hline
  \cellcolor{darkGray} & (1, 3) & $3.18$ & \textbf{0.04} &  0.96 & \textbf{16} & 26/20  & 10 & 2 & 0 & $\checkmark$ & 5 & \cellcolor{darkGray} \\ % 3.40 \\ 

 \cellcolor{darkGray} &  (2, 1) & $0.05$ & \textbf{0.04} & \textbf{0.04} & \textbf{2} & \textbf{2/0} & 1 & 1 & 0 & $\checkmark$ & 2 & \cellcolor{darkGray} \\ % 3.35 \\
 \cellcolor{darkGray} & (1, 3) & $3.10$ & 1.64 & \textbf{0.21} & $74$ & \textbf{2/10} & 10 & 2 & 0 & $\checkmark$ & 1 & \cellcolor{darkGray} \\ %3.03 \\
 \cellcolor{darkGray} & (1, 1) & \textbf{0.04} & 0.06 & 0.11 & \textbf{3} & 3/2 & 1 & 1 & 0 & $\checkmark$ & 1 & \cellcolor{darkGray} \\ %3.49 \\
 \cellcolor{darkGray} %& (3, 6) & $\perp$ & E & \textbf{833.3} & E & \textbf{1620/9199} & 100 & 20 & 40 & $\checkmark$ \\ 
  %& (3, 6) & $\perp$ & 1269 & \textbf{102.4} & 13706 & \textbf{1778/2960} & 100 & 20 & 40 & $\times$ \\ 
 \cellcolor{darkGray}   & (3, 6) & $\perp$ & 1269 & \textbf{112.5} & 13706 & \textbf{1170/4716} & 100 & 20 & 40 & $\times$ & 15 & \cellcolor{darkGray} \\ %4.86 \\ 

 %&(4, 6) & E & E & \textbf{4749.05} & E & \textbf{11604/166050} & 100 & 20 & 40 & $\times$ & & \\
 \cellcolor{darkGray} & (4, 5) & $\perp$ & $5251$ & \textbf{4144} & \textbf{44177} & 52623/12332 & 100 & 20 & 40 & $\times$ & 24 & \cellcolor{darkGray} \\ %4.87 \\
   %\cline{2-13}
 %& \multirow{2}{*}{(4, 5)} & \multirow{2}{*}{$\perp$} & \multirow{2}{*}{$5251$} & \multirow{2}{*}{\textbf{4144}} & \multirow{2}{*}{\textbf{44177}} & 52623 & \multirow{2}{*}{100} & \multirow{2}{*}{20} & \multirow{2}{*}{40} & \multirow{2}{*}{$\times$} & & \\
 %&  & & & & & 12332 & & & & & & \\
   %\cline{2-13}
 %& (2, 4) & E & E & \textbf{1377.67} & E & \textbf{3915/55770} & 100 & 20 & 40 & $\times$ \\
  \cellcolor{darkGray} & (3, 5) & $\perp$ & 2044 & \textbf{359.3} & 31363 & \textbf{9123/10158} & 100 & 20 & 40 & $\times$ & 9 & \cellcolor{darkGray} \\ % 3.21 \\
 \multirow{-8}{*}{\cellcolor{darkGray} {\textit{Tr.}}} & (4, 12) & $\perp$ & $\perp$ & \textbf{6571} & $\perp$ & \textbf{2728/40920} & 100 & 20 & 40 & $\times$ & 104 & \multirow{-8}{*}{\cellcolor{darkGray} {5.13}} \\ % 5.13 \\
 \hline \hline
 \cellcolor{darkGray}  \textit{Con.}& (2, 2) & 0.23 & \textbf{0.09} & \textbf{0.09} & \textbf{4} & \textbf{4/0} & 3 & 3 & 0 & $\checkmark$ & 4 & \cellcolor{darkGray} 4.37 \\
  \hline \hline
 \cellcolor{lightGray} \textit{Coo.}&  (3, 5) & $\perp$ & 1356 & \textbf{2.81} & 27883 & \textbf{16/160} & 20  & 2 & 0 & $\checkmark$ & 1 & \cellcolor{lightGray} 3.64 \\

 \hline \hline
 \cellcolor{lightGray} & (2, 3) & $3.40$ & 0.21 & \textbf{0.17} & \textbf{8} & \textbf{8/0} & 3 & 3 & 0 & $\checkmark$ & 8 & \cellcolor{lightGray} \\ % 3.93 \\ 
 %& (3, 5) & $\perp$ & E & \textbf{1390.75} & E & \textbf{18038}  & a & b & c & $\times$ & & \\
\multirow{-2}{*}{\cellcolor{lightGray} {\textit{Usb.}}}  & (3, 5) & $\perp$ & \textbf{231.9} & 364.4 & \textbf{5638} & \textbf{5638/0}  & 20 & 2 & 0 & $\checkmark$ & 32 & \multirow{-2}{*}{\cellcolor{lightGray} {3.93}} \\ % 4.40 \\
 \hline \hline
  \cellcolor{darkGray} & (8, 8) & $\perp$ & \textbf{18.19} & 18.20 & \textbf{256} & \textbf{256/0} & 40 & 2 & 0 & $\checkmark$ & 256 & \cellcolor{darkGray} \\ % 6.06 \\ 
  \multirow{-2}{*}{\cellcolor{darkGray} {\textit{St.}}} & (3, 6) & $\perp$ & 1311 & \textbf{194.8} & 14994 & \textbf{1697/6536}  & 100 & 20 & 40 & $\times$ & 45 & \multirow{-2}{*}{\cellcolor{darkGray} {6.06}} \\ % 5.43\\
 \hline \hline
   \multirow{7}{*}{\textit{Syn.}} & \cellcolor{darkGray} (2, 2) & 0.21 & 0.24 & \textbf{0.18} & 11 & \textbf{4/3}  & 3 & 3 & 0 & $\checkmark$ & 2 & \cellcolor{darkGray} 4.12 \\ 
 & \cellcolor{lightGray} (2, 3) & 3.42 & 2.69 & \textbf{1.24} & 119 & \textbf{14/40} & 10 & 2 & 0 & $\checkmark$ & 3 & \cellcolor{lightGray} 4.11\\
 & \cellcolor{lightGray} (2, 4) & $2842$ & 108.6 & \textbf{16.51} & 3982 & \textbf{188/620} & 10 & 2 & 0 & $\checkmark$ & 3 & \cellcolor{lightGray}  4.28 \\
 & \cellcolor{lightGray} (2, 5) & $\perp$ & $7151$ & \textbf{68.90} & $44259$ & \textbf{380/2800} & 20 & 2 & 0 & $\checkmark$ & 11 & \cellcolor{lightGray} 4.53 \\
 %SINCE THEY HAVE BEEN ABOUT 2 HOURS, but there was an error and could not see the precise number
  %& (2, 6) & $\perp$ & $\perp$ & \textbf{102.46} & E & \textbf{1778/2960} & a & b & c & $\checkmark$ \\ 
  %& (2, 6) & $\perp$ & $\perp$ & \textbf{793.35} & E & \textbf{1935/28590} & 30 & 2 & 0 & $\checkmark$ \\ %COLAB
  %& (2, 6) & $\perp$ & $\perp$ & \textbf{546.68} & E & \textbf{1463/27120} & 40 & 2 & 0 & $\checkmark$ \\ 
  & \cellcolor{lightGray} (2, 6) & $\perp$ & $\perp$ & \textbf{402.2} & $\perp$ & \textbf{4792/9941} & 100 & 20 & 40 & $\times$ & 24 & \cellcolor{lightGray} 4.85 \\  
  %& (2, 6) & $\perp$ & $\perp$ & \textbf{520.38} & $\perp$ & \textbf{6551/10029} & 100 & 20 & 40 & $\times$ \\  
  %& (2, 7) & $\perp$ & E & \textbf{7294.98} & E & \textbf{9320/172640}  & 40 & 2 & 0 & $\checkmark$ \\
  & \cellcolor{darkGray} (2, 7) & $\perp$ & $\perp$ & \textbf{3596} & $\perp$ & \textbf{7344/139440}  & 40 & 2 & 0 & $\checkmark$ & 1 & \cellcolor{darkGray} 5.30 \\
  & \cellcolor{darkGray} (2, 7) & $\perp$ & $\perp$ & \textbf{3862} & $\perp$ & \textbf{24311/40615}  & 200 & 20 & 40 & $\times$ & 45 & \cellcolor{darkGray}  5.99 \\
 \hline

\end{tabular}

\caption{Empirical evaluation results of the different Boolean abstraction algorithms%}
         , where the best results are in \textbf{bold} and $\phiB$ only refers to best times.}
  \label{tab:benchmark}
  % \end{table}
\end{figure}
} % load benchmark tables to be used later by a newcommand call '\TableBenchmark'

\begin{abstract} 
  In this paper, we address the problem of the (reactive) realizability of
  %safety 
  specifications of theories richer than Booleans, including
  arithmetic theories.
  Our approach transforms theory specifications into purely Boolean
  specifications by (1) substituting theory literals by Boolean
  variables, and (2) computing an additional Boolean requirement that
  captures the dependencies between the new variables imposed by the
  literals.
  The resulting specification can be passed to existing Boolean
  off-the-shelf realizability tools, and is realizable if and only if the
  original specification is realizable.
  The first contribution is a brute-force version of our
  method, which requires a number of SMT queries that is doubly
  exponential in the number of input literals.
  Then, we present a faster method that exploits a nested encoding of
  the search for the extra requirement and uses SAT solving for faster
  traversing the search space and uses SMT queries internally.
  Another contribution is a prototype in Z3-Python.
  Finally, we report an empirical evaluation using specifications
  inspired in real industrial cases.
  To the best of our knowledge, this is the first method that succeeds in non-Boolean LTL realizability.
\end{abstract}

%%% Local Variables:
%%% TeX-master: "../main.tex"
%%% TeX-PDF-mode: t
%%% End:

\section{Introduction}
\label{sec:intro}

%
% 1. Problem and solution
%
Reactive
synthesis~\cite{pnueli89onthesythesis,pnueli89onthesythesis:b} is the
problem of automatically producing a system that is guaranteed to
model a given temporal specification, where the Boolean variables
(i.e., atomic propositions) are split into variables controlled by the
environment and variables controlled by the system.
Realizability is the related decision problem of deciding whether such
a system exists.
These problems have been widely
studied~\cite{jacobs17reactive,finkbeiner13bounded}, specially in the
domain of Linear Temporal Logic (LTL)~\cite{pnueli77temporal}.
Realizability corresponds to infinite games where players
alternatively choose the valuations of the Boolean variables they
control.
The winning condition is extracted from the temporal specification and
determines which player wins a given play.
A system is realizable if and only if the system player has a winning
strategy, i.e., if there is a way to play such that the
specification is satisfied in all plays played according to the
strategy.

However, in practice, many real and industrial specifications use
complex data beyond Boolean atomic propositions, which precludes the
direct use of realizability tools.
These specifications cannot be written in (propositional) LTL, but
instead use literals from a richer domain.
We use $\LTLt$ for the extension of LTL where Boolean atomic
propositions can be literals from a (multi-sorted) first-order theory
$\mathcal{T}$.
The $\calT$ variables (i.e., non-Boolean) in the specification are again
split into those controlled by the system and those controlled by the
environment.
The resulting realizability problem also corresponds to infinite
games, but, in this case, players chose valuations from the domains of
$\calT$, which may be infinite.
Therefore, arenas may be infinite and positions may have infinitely
many successors.
In this paper, we present a method that transforms a specification
that uses data from a theory $\calT$ into an equi-realizable Boolean
specification.
The resulting specification can then be processed by an off-the-shelf
realizability tool.

The main element of our method is a novel \emph{Boolean abstraction}
method, which allows to transform \LTLt
specifications into pure (Boolean) LTL specifications.
The method first substitutes all $\calT$ literals by fresh Boolean
variables controlled by the system, and then extends the specification
with an additional sub-formula that constrains the combination values
of these variables.
This method is described in Section~\ref{sec:booleanAbs}.
The main idea is that, after the environment selects values for its
(data) variables, the system responds with values for the variables it
controls, which induces a Boolean value for all the literals.
The additional formula we compute captures the set of possible
valuations of literals and the precise power of each player to produce
each valuation.
%
% To simplify the treatment (and avoid dealing with finite traces) we
% assume reactive systems have no deadlock: if necessary, a blocking
% state is extended into a self-loop.

\begin{example}
  \label{example1}
  Consider the following specification
  $\varphi = \square (R_0 \wedge R_1)$, where:
  \[ R_0 : (x<2) \Into \Next(y>1) \hspace{4em}
    R_1 : (x \geq 2) \Into (y<x)
  \]
  where $x$ is a numeric variable that belongs to the environment and
  $y$ to the system.
  In the game corresponding to this specification, each player has an
  infinite number of choices at each time step.
  For example, in $\ThZ$ (the theory of integers), the environment
  player chooses an integer for $x$ and the system responds with an
  integer for $y$.
  This induces a valuation of all literals in the formula, which in
  turn induces (also considering the valuations of the literals at
  other time instants, according to the temporal operators) a
  valuation of the full specification.

  In this paper, we exploit that, from the point of view of the
  valuations of the literals, there are only \emph{finitely many}
  cases and provide a systematic manner to compute these cases.
  This allows us to reduce a specification into a purely Boolean
  specification that is equi-realizable.
  This specification encodes the (finite) set of decisions of the
  environment, and the (finite) set of reactions of the system. \qed
% With $n=2$ literals, there are $\mathcal{C}=2^n$ possible combinations
% of Boolean valuation literals, for example
% $c_0 = (x<2) \wedge (y>1) \wedge (y<x)$ and
% $c_1 =(x<2) \wedge (y>1) \wedge (y \geq x)$.
%
% We call each of these a \configuration.
%
% We can capture those values of the variables of the environment that
% allow the system to chose variables and produce a given configuration
% $c$ by $\exists\overline{y}.c(\overline{x},\overline{y})$.
%
% Similarly, $\forall\ybar.\neg{}c(\xbar,\ybar)$ caputures those values
% of $\xbar$ for which $c$ is not a possible response of the system.
% %
% A \reaction is a the conjunction of one of these choices for
% each \configuration, and it is a formula that captures those values of
% $\xbar$ for which the system has a specific power to choose among a
% precise set of \configuration as outcome.
% %
% The validity of the resulting formula (where the environment variable
% $\xbar$ is exeistentially quantified) corresponds to the existence of
% moves of the environment that leaves the system with the specific
% power described by the \reaction.

% By enumerating all $2^{\calC}$ possible reactions we cover all the
% different possible powers to choose an outcome that the environment
% can lead the system.
% %
% This brute force method traverses explicitly the set of \reactions,
% which doubly exponential in the number of literals.
% %
% Note that each reaction corresponds to the validity of a formula of
% the form $\exists^*\forall^*$ in theory $\calT$.
%
%\qed
\end{example}

Ex.~\ref{example1} suggests a naive algorithm to capture the powers of the
environment and system to determine a combination of the valuations of
the literals, by enumerating all these combinations and checking the
validity of each potential reaction.
Checking that a given combination is a possible reaction requires an
$\exists^*\forall^*$ query (which can be delegated to an SMT solver
for appropriate theories).

In this paper, we describe and prove correct a Boolean abstraction
method based on this idea.
Then, we propose a more efficient search method for the set of
possible \reactions using SAT solving to speed up the exploration of
the set of reactions.
The main idea of this faster method is to learn from an invalid
reaction which other reactions are guaranteed to be invalid, and from
a valid reaction which other reactions are not worth being explored.
We encode these learnt sets as a incremental SAT formula that allows
to prune the search space.
The resulting method is much more efficient than brute-force
enumeration because, in each iteration, the learning can prune an
exponential number of cases.
An important technical detail is that computing the set of cases to be
pruned from the outcome of a given query can be described efficiently
using a SAT solver.

In summary, our contributions are: (1) a proof that realizability is
decidable for all \LTLt specifications for those theories $\calT$ with
a decidable $\exists^*\forall^*$ fragment; (2) a simple implementation
of the resulting Boolean abstraction method; (3) a much faster method
based on a nested-SAT implementation of the Boolean
abstraction method that efficiently explores the search space of
potential reactions; and (4) an empirical evaluation of these
algorithms, where our early findings suggest that Boolean abstractions
can be used with specifications containing different arithmetic
theories, and also with industrial specifications.
We used Z3~\cite{demoura08z3} both as an SMT solver and a SAT solver,
and Strix~\cite{meyer18strix} as the realizability checker.
To the best of our knowledge, this is the first method that succeeds
(and efficiently) in non-Boolean LTL realizability.

%%% Local Variables:
%%% TeX-master: "../main.tex"
%%% TeX-PDF-mode: t
%%% End:

\section{Preliminaries} 
\label{sec:prelims}

We study realizability of LTL~\cite{pnueli77temporal,manna95temporal}
specifications.
The syntax of LTL is:
\[
  \varphi  ::= T \DefOR a \DefOR \varphi \lor \varphi \DefOR \neg \varphi
  \DefOR \Next \varphi \DefOR \varphi \U\varphi
\]
where $a$ ranges from an atomic set of proposition $\AP$, $\lor$,
$\land$ and $\neg$ are the usual Boolean disjunction, conjunction and
negation, and $\Next$ and $\U$ are the next and until temporal
operators.
The semantics of LTL associate traces $\sigma\in\Sigma^\omega$ with
formulae as follows:
 \[
   \begin{array}{l@{\hspace{0.3em}}c@{\hspace{0.3em}}l}
     \sigma \models T & \text{always} & \\
     \sigma \models a & \text{iff } & a \in\sigma(0) \\
     \sigma \models \varphi_1 \Or \varphi_2 & \text{iff } & \sigma\models \varphi_1 \text{ or } \sigma\models \varphi_2 \\
%     \sigma \models \varphi_1 \And \varphi_2 & \text{iff } & \sigma\models \varphi_1 \text{ and } \sigma\models \varphi_2 \\
     \sigma \models \neg \varphi & \text{iff } & \sigma \not\models\varphi \\
     \sigma \models \Next \varphi & \text{iff } & \sigma^1\models \varphi \\
     \sigma \models \varphi_1 \U \varphi_2 & \text{iff } & \text{for some } i\geq 0\;\; \sigma^i\models \varphi_2, \text{ and } \text{for all } 0\leq j<i, \sigma^j\models\varphi_1 \\
   \end{array}
 \]
 We use common derived operators like $\vee$, $\calR$, $\Event$ and
 $\Always$.

Reactive synthesis
\cite{thomas08church,piterman06synthesis,bloem12synthesis,finkbeiner16synthesis,bloemETAL2021vacuitySynthesis}
is the problem of producing a system from an LTL specification, where
the atomic propositions are split into propositions that are
controlled by the environment and those that are controlled by the
system.
Synthesis corresponds to a turn-based game where, in each turn, the
environment produces values of its variables (inputs) and the system
responds with values of its variables (outputs).
A play is an infinite sequence of turns.
The system player wins a play according to an LTL formula $\varphi$ if
the trace of the play satisfies $\varphi$.
A (memory-less) strategy of a player is a map from positions into a
move for the player.
A play is played according to a strategy if all the moves of the
corresponding player are played according to the strategy.
A strategy is winning for a player if all the possible plays played
according to the strategy are winning.

Depending on the fragment of LTL used, the synthesis problem has
different complexities.
The method that we present in this paper generates a formula in the
same temporal fragment as the original formula (e.g., starting
from a safety formula another safety formula is generated).
The generated formula is discharged into a solver capable to solve
formulas in the right fragment.
For simplicity in the presentation, we illustrate our method with
safety formulae.

We use $\LTLt$ as the extension of LTL where propositions are replaced
by literals from a first-order theory $\calT$.
In realizability for $\LTLt$, the variables that occur in the literals
of a specification $\varphi$ are split into those variables controlled
by the environment (denoted by $\overline{v}_e$) and those controlled
by the system ($\overline{v}_s$), where $\VE \cap \VS = \emptyset$.
We use $\varphi(\VE,\VS)$ to remark that $\VE\cup\VS$ are the
variables occurring in $\varphi$.
The alphabet $\Sigma_{\calT}$ is now a valuation of the variables in
$\VE\cup\VS$.
A trace is an infinite sequence of valuations, which induces an
infinite sequence of Boolean values of the literals occurring in
$\varphi$ and, in turn, a valuation of the temporal formula.

Realizability for $\LTLt$ corresponds to an infinite game with an
infinite arena where positions may have infinitely many successors
if the ranges of the variables controlled by the system and the
environment are infinite.
For instance, in Ex.~\ref{example1} with $\calT=\ThZ$, valuation ranges
over infinite values, and literal $(x \geq 2)$ can be satisfied with $x=2$,
$x=3$, etc.

Arithmetic theories are a particular class of first-order
theories.
Even though our Boolean abstraction technique is applicable to
any theory with a decidable $\exists^*\forall^*$ fragment, we
illustrate our technique with arithmetic specifications.
Concretely, we will consider $\calT_{\mathbb{Z}}$ (i.e., linear
integer arithmetic) and $\calT_{\mathbb{R}}$ (i.e., non-linear real
arithmetic).
Both theories have a decidable $\exists^*\forall^*$ fragment.
Note that the choice of the theory influences the realizability of a given
formula.

\begin{example}
  \label{ex:two}
  Consider Ex.~\ref{example1}. The formula $\varphi:=R_0\And R_1$ is
  not realizable for $\calT_{\mathbb{Z}}$, since, if at a given
  instant $t$, the environment plays $x=0$ (and hence $x<2$ is true),
  then $y$ must be greater than $1$ at time $t+1$. Then, if at $t+1$
  the environment plays $x=2$ then $(x\geq 2)$ is true but there is no
  $y$ such that both $(y>1)$ and $(y<2)$.
  However, for $\calT_{\mathbb{R}}$, $\varphi$ is realizable (consider
  the system strategy to always play $y=1.5$).

  The following slight modifications of Ex.~\ref{example1} alters its
  realizability ($R_1'$ substitutes $R_1$ by having the $\calT$-predicate $y{}\leq x$ instead of $y<x$):
  \[
    R_0 : (x<2) \Into \Next(y>1) \hspace{4em}R_1' : (x \geq 2) \Impl (y \leq x)
  \]
  Now, $\varphi' = \Always(R_0 \wedge R_1')$ is realizable for both
  $\calT_{\mathbb{Z}}$ and $\calT_{\mathbb{R}}$, as the strategy of
  the system to always pick $y=2$ is winning in both theories.\qed
\end{example}  

%%% Local Variables:
%%% TeX-master: "../main.tex"
%%% TeX-PDF-mode: t
%%% End:

\section{Boolean Abstraction} \label{sec:booleanAbs}

\newcommand{\vs}{\ensuremath{\overline{v}}\xspace}
\newcommand{\ws}{\ensuremath{\overline{w}}\xspace}

We solve the realizability problem modulo theories by transforming the
specification into an equi-realizable Boolean specification.
Given a specification $\varphi$ with literals $l_i$, we get a
new specification
$\varphi[l_i \leftarrow s_i] \wedge \Always\phiEx$, where $s_i$ are
fresh Boolean variables and $\phiEx \in \LTLB$ is a Boolean
formula (without temporal operators).
The additional sub-formula $\phiEx$ uses the freshly introduced
variables $s_i$ controlled by the system, as well as additional
Boolean variables controlled by the environment $\overline{e}$, and
captures the precise combined power of the players to decide the
valuations of the literals in the original formula.
We call our approach \emph{Booleanization} or \emph{Boolean
  abstraction}.
The approach is summarized in Fig.~\ref{figBigPicture}: given an LTL
specification $\phiT$, it is translated into a Boolean
$\phiB$ which can be analyzed with off-the-shelf
realizability checkers.
Note that $\GameB$ and $\GameT$ are the games constructed from specifications
$\phiB$ and $\phiT$, respectively. 
Also, note that \cite{gradelETAL2002automataLogicsInfiniteGames} shows that we can
construct a game $\mathcal{G}$ from a specification $\varphi$ and that
$\varphi$ is realizable if and only if $\mathcal{G}$ is winning for the
system.

\begin{figure}[b!]
  \centering
  \includegraphics[scale=0.35]{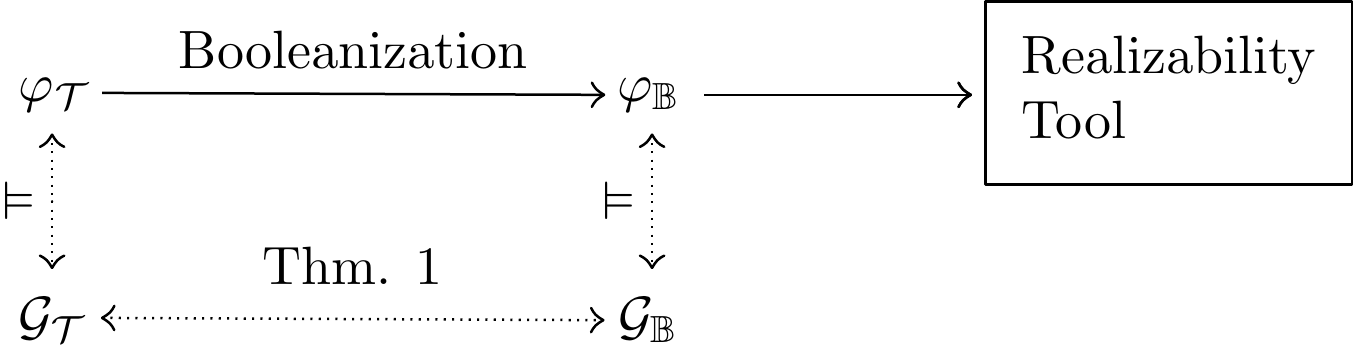}
  \caption{The tool chain with the correctness argument.}
  \label{figBigPicture}
\end{figure}

The Booleanization procedure constructs an extra requirement $\phiEx$
and conjoins $\Always\phiEx$ with the formula
$\varphi[l_i \leftarrow s_i]$.
In a nutshell, after the environment chooses a valuation of the
variables it controls (including $\overline{e}$), the system responds
with valuations of its variables (including $s_i$), which induces a
Boolean value for all literals.
Therefore, for each possible choice of the environment, the system has
the power to choose a Boolean response among a specific collection of
responses (a subset of all the possible combinations of Boolean
valuations of the literals).
Since the set of all possible responses is finite, so are the
different cases.
The extra requirement captures precisely the finite collection of
choices of the environment and the resulting finite collection of
responses of the system for each case.

\subsection{Notation}

In order to explain the construction of the extra requirement, we
introduce some preliminary definitions.
We will use Ex.~\ref{example1} as the running example.

%\subsubsection{Literals.}
%
A literal is an atom or its negation, 
regardless of whether the atom is a Boolean variable or a predicate of a theory.
Let $\Lit(\varphi)$ be the collection of literals that appear in
$\varphi$ (or $\Lit$, if the formula is clear from the context).
For simplicity, we assume that all literals belong the same theory,
but each theory can be Booleanized in turn, as each literal belongs to
exactly one theory and we assume in this paper that literals from
different theories do not share variables.
We will use \xs as the environment controlled variables occurring in
$\Lit(\varphi)$ and \ys for the variables controlled by the system.

In Ex.~\ref{example1}, we first translate the literals in $\varphi$.
Since $(x<2)$ is equivalent to $\neg (x \geq 2)$, we use a single
Boolean variable for both.
The substitutions is:
\[
  \begin{array}{r@{\;\leftarrow\;}l@{\hspace{3em}}r@{\;\leftarrow\;}l@{\hspace{3em}}r@{\;\leftarrow\;}l}
    (x<2) & s_0 & (y>1) & s_1 & (y<x)& s_2 \\
    (x\geq{}2) & \neg{}s_0 & (y\leq{}1) & \neg{}s_1 & (y\geq x)& \neg{}s_2 \\
  \end{array}
\]
%
% \begin{compactitem}
%     \item $(x<2)\leftarrow{}s_0$ and $(x \geq 2)\leftarrow{}\neg{}s_0$
%     \item $(y>1)\leftarrow{}s_1$ and $(y \leq 1)\leftarrow{}\neg{}s_1$
%     \item $(y<x)\leftarrow{}s_2$ and $(y \geq x)\leftarrow{}\neg{}s_2$
% \end{compactitem}
%
After the substitution we obtain
$\varphi'' = \Always(R_0^{\mathbb{B}} \wedge R_1^{\mathbb{B}})$ where
\[
R_0^{\mathbb{B}}: s_0 \Impl \Next s_1 \hspace{4em}
R_1^{\mathbb{B}}: \neg s_0 \Impl s_2
\]
Note that $\varphi''$ may not be equi-realizable to $\varphi$, as we
may be giving too much power to the system if $s_0$, $s_1$ and $s_2$
are chosen independently without restriction.
Note that $\varphi''$ is realizable, for example by always choosing
$s_1$ and $s_2$ to be true, but $\varphi$ is not realizable in
$\textit{LTL}_{\ThZ}$. This justifies the need of an extra sub-formula.

%\subsubsection{Choices}
%
\begin{definition}[Choice]
  A \configuration $c\subseteq \Lit(\varphi)$ is a subset of the literals of $\varphi$.
\end{definition}
The intended meaning of a \configuration is to capture what literals
are true in the \configuration, while the rest
(i.e., $\Lit\setminus{}c$) are false.
Once the environment picks values for $\xs$, the system can realize
some \configuration $c$ by selecting $\ys$ and making the literals in
$c$ true (and the rest false).
However, for some values of $\xs$, some \configurations may not be
possible for the system for any $\ys$.
Given a \configuration $c$, we use $f(c(\xs,\ys))$ to denote the formula:
\[ \bigwedge_{l\in c} l \wedge \bigwedge_{l \notin c} \neg l \]
which is a formula with variables $\xs$ and $\ys$ that captures
logically the set of values of $\xs$ and $\ys$ that realize precisely
\configuration $c$.
We use $\calC$ for the set of \configurations.
Note that there are $|\calC|=2^{|\Lit|}$ different \configurations.
We call the elements of $\calC$ \configurations because they may be at
the disposal of the system to choose by picking the right values of its variables.

% \begin{center} \label{tablArithmeticDecidability}
% \begin{tabular}{||c c c c c||} 
%  \hline
%  $\#$ & $\mathbb{B}$ & $l_0$ & $l_1$ & $l_2$ \\ [0.5ex] 
%  \hline\hline
%  $c_0$ & $s_0 \wedge s_1 \wedge s_2$ & $(x<2)$ & $(y>1)$ & $(y<x)$\\ 
%  \hline
%  $c_1$ & $s_0 \wedge s_1 \wedge \neg s_2$ & $(x<2)$ & $(y>1)$ & $(y %\geq x)$\\
%  \hline
%  ... & ... & ... & ... & ...\\
% \hline
%  $c_7$ & $\neg s_0 \wedge \neg s_1 \wedge \neg s_2$ & $(x \geq 2)$ & %$(y \leq 1)$ & $(y \geq x)$\\
%  \hline
% \end{tabular}
% \end{center}

%\subsubsection{Potentials and antipotentials.} 
%
A given \configuration $c$ can act as \textit{potential}
(meaning that the response is possible) or as \textit{antipotential}
(meaning that the response is not possible).
A potential is a formula (that depends only on \xs) that captures
those values of $\xs$ for which the system can respond and make
precisely the literals in $c$ true (and the rest of the literals
false).
The negation of the potential (i.e., an antipotential) captures precisely those values of $\xs$
for which there are no values of $\ys$ that lead to $c$.

\begin{definition}[Potential and Antipotential]
  Given a choice $c$, a potential is the following formula $\PT{c}$ and an
  antipotential is the following formula $\NT{c}$:
\[
  \PT{c}(\xs) = \exists\ys.f(c(\xs,\ys))\hspace{5em}
  \NT{c}(\xs) = \forall\ys.\neg{}f(c(\xs,\ys))
\]
\end{definition}

\begin{example}
  \label{ex:three}
  We illustrate two choices for Ex.~\ref{example1}.
  Consider choices $c_0=\{(x<2), (y>1), (y<x)\}$ and
  $c_1=\{(x<2),(y>1)\}$.
    Choice $c_0$ corresponds to
    $f(c_0) = (x<2) \wedge (y>1) \wedge (y<x)$, that is, literals
    $(x<2)$, $(y>1)$ and $(y<x)$ are true.
    Choice $c_1$ corresponds to
    $f(c_1) =(x<2) \wedge (y>1) \wedge (y \geq x)$, that is,
    literals $(x<2)$ and $(y>1)$ being true and $(y < x)$ being false
    (i.e., $(y \geq x)$ being true).
It is easy to see the meaning of $c_2$, $c_3$ etc.
Then, the 
    potential and antipotential formulae of e.g., \configurations $c_0$
    and $c_1$ from Ex.~\ref{example1} are as follows:
\[
\begin{array}{l@{\hspace{2.3em}}l}
  \PT{c_0} =  \exists y . (x<2) \And (y>1) \And (y<x) & \NT{c_0} = \forall y . \neg\big( (x<2) \And (y>1) \And (y<x)\big) \\
\PT{c_1} = \exists y .  (x<2) \And (y>1) \And (y \geq x) &
  \NT{c_1} = \forall y . \neg\big((x<2) \And (y>1) \And (y \geq x)\big)
\end{array}
\]
Note that potentials and antipotentials have $\xs$ as the only free variables.
\qed
\end{example}
Depending on the theory, the validity of potentials and antipotentials
may be different.
For instance, consider $\PT{c_0}$ and theories $\mathcal{T}_{\mathbb{Z}}$ and
 $\mathcal{T}_{\mathbb{R}}$:
\begin{itemize}
%\begin{compactitem}
\item In $\mathcal{T}_{\mathbb{Z}}$:
  $\exists y . (x<2) \wedge (y>1) \wedge (y<x)$ is
  equivalent to \textit{false}.
\item In $\mathcal{T}_{\mathbb{R}}$:
  $\exists y . (x<2) \wedge (y>1) \wedge (y<x)$ is equivalent to
  $(x<2)$.
  % \end{compactitem}
\end{itemize}
These equivalences can be obtained using classic quantifier elimination procedures, e.g., with Cooper's
algorithm~\cite{cooper1972theoremProving} for
$\mathcal{T}_{\mathbb{Z}}$ and Tarski's
method~\cite{tarski1951decisionElementary} for
$\mathcal{T}_{\mathbb{R}}$.
%
% We will only focus on $\mathcal{T}_{\mathbb{Z}}$ now.
  
%\subsubsection{Reactions and valid reactions.}
%
A reaction is a description of the specific \configurations that the
system has the power to choose.
\begin{definition}[Reaction] 
  Let $P$ and $A$ be a partition of $\calC$ that is:
  $P \subseteq \calC$, $A\subseteq \calC$, $P\cap A=\emptyset$ and
  $P\cup A=\calC$.
  The reaction $\react_{(P,A)}$ is as follows:
\[
  \react_{(P,A)}(\xs) \DefinedAs \bigwedge_{c \in P} \PT{c} \wedge
  \bigwedge_{c \in A} \NT{c}
\]
\end{definition}
The reaction $\react_{(P,A)}$ is equivalent to:
\[
  \react_{(P,A)}(\xs) = \bigwedge_{c \in P} \big(\exists \ys. f(c(\xs,\ys))\big) \wedge \bigwedge_{c \in A} \big(\forall \ys.\neg f(c(\xs,\ys))\big).
\]
There are $2^{2^{|\Lit|}}$ different reactions.

A reaction $r$ is called valid whenever there is a move of the
environment for which $r$ captures precisely the power of the system,
that is exactly which \configurations the system can choose.
Formally, a reaction is valid whenever $\exists \xs. r(\xs)$ is a
valid formula.
We use $\calR$ for the set of reactions and $\textit{VR}$ for the set of valid
reactions.
It is easy to see that, for all possible valuations of $\xs$ the
environment can pick, the system has a specific power to respond
(among the finitely many cases).
Therefore, the following formula is valid:
\[
  \varphi_{\textit{VR}}=\forall\xs.\bigvee_{r\in\textit{VR}}r(\xs).
\]
\begin{example}
  \label{ex:four}
In Ex.~\ref{example1}, for theory $\mathcal{T}_{\mathbb{Z}}$, we find there are two valid reactions
(using choices from Ex.~\ref{ex:three}):
\[
  \begin{array}{rcl}
    r_1 & : & \exists x .c_0^a \And c_1^p \And c_2^p \And c_3^p \And
              c_4^a \And c_5^a \And c_6^a \And c_7^a \\
    r_2 & : & \exists x . c_0^a \And c_1^a \And c_2^a \And c_3^a
              \And c_4^a \And c_5^p \And c_6^p \And c_7^a,
  \end{array}
\]
%              
% \begin{compactitem}
% \item
%   $r_i : \exists x .c_0^a \And c_1^p \And c_2^p \And c_3^p \And
%   c_4^a \And c_5^a \And c_6^a \And c_7^a$
% \item
%   $r_j : \exists x . c_0^a \And c_1^a \And c_2^a \And c_3^a
%   \And c_4^a \And c_5^p \And c_6^p \And c_7^a,$
% \end{compactitem}
% %
where reaction $r_1$ models the possible responses of the system after the
environment picks a value for $x$ with $(x<2)$, whereas $r_2$ models
the responses to $(x \geq 2)$. On the other hand, for $\mathcal{T}_{\mathbb{R}}$, there are three
valid reactions:
\[
  \begin{array}{rcl}
    r_1 & : & \exists x. c_0^a \wedge c_1^p \wedge c_2^p \wedge
  c_3^p \wedge c_4^a \wedge c_5^a \wedge c_6^a \wedge c_7^a \\
    r_2 & : & \exists x. c_0^p \wedge c_1^p \wedge c_2^p 
    \wedge c_3^a \wedge c_4^a \wedge c_5^a \wedge c_6^a \wedge c_7^a \\
    r_3 & : & \exists x . c_0^a \wedge c_1^a \wedge c_2^a 
    \wedge c_3^a \wedge c_4^p \wedge c_5^p \wedge c_6^p \wedge c_7^a
  \end{array}
\]

Note that there is one valid reaction more, since in $\mathcal{T}_{\mathbb{R}}$ 
there is one more case: $x\in(1,2]$.
Also, note that $c_4$ cannot be a potential in
$\mathcal{T}_{\mathbb{Z}}$ (not even with a collaboration between
environment and system), whereas it can in $\mathcal{T}_{\mathbb{R}}$.
\qed
\end{example}

\subsection{The Boolean Abstraction Algorithm}
\label{sec:boolabsalg}

Boolean abstraction is a method to compute $\phiB$ from $\phiT$.
In this section we describe and prove correct a basic brute-force
version of this method, and later in Section~\ref{sec:efficient}, we
present faster algorithms.
%
%Alg.~\ref{algoCommonFactor} shows the \textit{skeleton} for all versions
%of the Boolean abstraction method.
%
%All instances first calculate potentials and antipotentials, and then
All Boolean abstraction algorithms that we present on this paper first
compute the extra requirement, by visiting the set of reactions and
computing a subset of the valid reactions that is sufficient to
preserve realizability.
The three main building blocks of our algorithms are
(1) the stop criteria of the search for reactions;
(2) how to obtain the next reaction to consider;
and (3) how to modify the current set of valid reactions
(by adding new valid reactions to it) and the set of remaining
reactions (by pruning the search space).
Finally, after the loop, the algorithm produces as $\phiExtra$ a
conjunction of cases, one per valid reaction $(P,A)$ in $\textit{VR}$.
\begin{wrapfigure}[12]{l}{0.46\textwidth}
  \begin{minipage}{0.46\textwidth}
    \vspace{-2.3em}
\begin{algorithm}[H]
  \AlgBruteForce
 \caption{Brute-force}
 \label{algoBruteForce}
\end{algorithm}
\end{minipage}
%\vspace{-2em}
\end{wrapfigure}
We introduce a fresh variable $e_{(P,A)}$, controlled by the
environment for each valid reaction $(P,A)$, to capture that the
environment plays values for $\xbar$ that correspond to the case where
the system is left with the power to choose captured precisely by
$(P,A)$.
Therefore, there is one additional environment Boolean variable per
valid reaction (in practice we can enumerate the number of valid
reactions and introduce only a logarithmic number of environment
variables).
Finally, the extra requirement uses $P$ for each valid reaction
$(P,A)$ to encode the potential moves of the systems as a disjunction
of the literals described by each \configuration in $P$.
Each of these disjunction contains precisely the combinations of
literals that are possible for the concrete case that $(P,A)$
captures.

A brute-force algorithm that implements Boolean abstraction method by
exhaustively searching all reactions is shown in
Alg~\ref{algoBruteForce}.
The building blocks of this algorithm are:
%
%\begin{enumerate}[(1)]
\begin{compactenum}[(1)]
\item It stops when the remaining set of reactions is empty.
\item It traverses the set $\calR$ according to some predetermined order.
\item To modify the set of valid reactions, if $(P,A)$ is valid it
  adds $(P,A)$ to the set \textit{VR} (line $7$). To modify the set of
  remaining reactions, it removes $(P,A)$ from the search.
\end{compactenum}
%\end{enumerate}
%
Finally, the extra sub-formula $\phiExtra$ is generated by
\textit{getExtra} (line $8$) defined as follows:
\[
  \textit{getExtra}(\VR) =
  \bigwedge_{(P,A)\in \textit{VR}} (e_{(P,A)} \rightarrow \bigvee_{c\in P} (\bigwedge_{l_i\in c}s_i\And \bigwedge_{l_i\notin c}\neg s_i))
\]
Note that there is an $\exists^*\forall^*$ validity query in the body of the loop (line 6) 
to check whether the candidate reaction is valid.
This is why decidability of the $\exists^*\forall^*$ fragment is crucial because it captures the finite partitioning of the environment moves 
(which is existentially quantified) for which the system can react in certain ways (i.e., potentials, which are existentially quantified) 
by picking appropriate valuations but not in others (i.e., antipotentials, which are universally quantified).
In essence, the brute-force algorithm iterates over all the reactions,
one at a time, checking whether each reaction is valid or not.
In case the reaction (characterized by the set of potential
\configurations\footnote{The potentials in a choice characterize the
  precise power of the system player, because the potentials
  correspond with what the system can respond.}) is valid, it is
added to $\VR$.

\begin{example}
  \label{ex:five}
  Consider again the specification in Ex.~\ref{example1}, with $\ThZ$
  as theory. Note that the valid reactions are $r_1$ and $r_2$, as
  shown in Ex.~\ref{ex:four}, where the potentials of $r_1$ are
  $\{c_1,c_2,c_3\}$ and the potentials of $r_2$ are $\{c_5,c_6\}$.
  Now, the creation of $\phiEx$ requires two fresh variables $d_0$ and
  $d_1$ for the environment (they correspond to environment decisions
  $(x<2)$ and $(x \geq 2)$, respectively), resulting into:
 \begin{align*}
   \varphi^{\textit{extra}}_{\mathcal{T}_{\mathbb{Z}}} :
   \begin{pmatrix}
     \begin{array}{rcl}
   d_0 & \Impl & \big(     (s_0 \And s_1 \And \neg s_2) \Or
     (s_0 \And \neg s_1 \And s_2) \Or
       (s_0 \And \neg{}s_1 \And \neg s_2)
       \big)
 \\  &\And& \\
 d_1 &\Impl&\big( (\neg{}s_0 \And s_1 \And \neg s_2 )\Or 
       (\neg{}s_0 \And \neg{}s_1 \And s_2)\big)
     \end{array}
   \end{pmatrix}
 \end{align*}

 For example $c_2=\{s_0\}$ is a choice that appears as potential in
 valid reaction $r_1$, so it appears as a disjunct of $d_0$ as
 $(s_0\And\neg{}s_1\And\neg{}s_2)$.
 The resulting \textit{Booleanized} specification $\phiB$ is as follows:
\[
  \hspace{10.6em}
  \varphi^{\mathbb{B}}_{\mathcal{T}_{\mathbb{Z}}} = (\varphi'' \And 
  \square (A_{\mathbb{B}} \Impl \varphi^{\textit{extra}}_{\mathcal{T}_{\mathbb{Z}}}))
  \hspace{10.6em}\qed
\]
%\qed
\end{example}

Note that the Boolean encoding is extended with an assumption formula 
$A_{\mathbb{B}} = (d_0 \leftrightarrow \neg d_1) \wedge (d_0 \vee d_1)$
that restricts environment moves to guarantee that exactly one
environment decision variable is picked.
Also, note that a Boolean abstraction algorithm will output three (instead of two)
decisions for the environment, but we ackowledge that one of them will never be played by 
it, since it gives strictly more power to the system.
The complexity of this brute-force Booleanization algorithm is doubly
exponential in the number of literals.

\subsection{From Local Simulation to Equi-Realizability} 

The intuition about the correctness of the algorithm is that the extra
requirement encodes precisely all reactions (i.e., collections of
\configurations), for which there is a move of the environment that
leaves the system with precisely that power to respond.
%
%%%%%LOCAL PROOF
As an observation, in the extra requirement, the set of potentials in
valid reactions cannot be empty.
%
%In other words, for every move of the environment the system can
%always move with a valid reaction, which will result in the
%always-existence of some outcome.
%
This is stated in Lemma \ref{lemmNotEmptyReactions}.

\begin{lemma} \label{lemmNotEmptyReactions} Let $C \in \mathcal{C}$ be
  such that $react_C \in \textit{VR}$. Then $C \neq \emptyset$.
\end{lemma}

\begin{proof}
  Bear in mind $\react_C \in \textit{VR}$ is valid. 
  Let $\vs$ be such that $\react_C[\xs \shortleftarrow \vs]$ is
  valid.
  Let $\ws$ be an arbitrary valuation of $\ys$ and let $c$ be a
  choice and $l$ a literal.
  Therefore:
  \[\bigwedge_{l[\xs \shortleftarrow \vs, \ys \shortleftarrow \ws]\textit{ is true }}l \wedge \bigwedge_{l[\xs \shortleftarrow \vs, \ys \shortleftarrow \ws]\textit{ is false }}\neg l\]
  It follows that $I[\xs \leftarrow \vs] \exists \ys . c$, so $c\in C$.
  \qed
\end{proof}
Lemma~\ref{lemmNotEmptyReactions} is crucial, because it ensures that
once a Boolean abstraction algorithm is executed, for each fresh $\overline{e}$ variable
in the extra requirement, at least one reaction with one or more
potentials can be responded by the system.

Therefore, in each position in the realizability game, the system can
respond to moves of the system leaving to precisely corresponding
positions in the Boolean game.
In turn, this leads to equi-realizability because each move can be
simulated in the corresponding game.
Concretely, it is easy to see that we can define a simulation between the positions of the games for $\phiT$
and $\phiB$ such that (1) each literal $l_i$ and the corresponding
variable $s_i$ have the same truth value in related positions, (2) the
extra requirement is always satisfied, and (3) moves of the system in
each game from related positions in each game can be mimicked in the other
game.
This is captured by the following theorem:

\begin{theorem} 
  \label{temporalGlobalTheorem} 
   System wins $\GameT$ if and only if System wins the game $\GameB$.
  Therefore, $\phiT$ is realizable if and only if $\phiB$ is
  realizable.
\end{theorem}

\begin{proof}
(Sketch). 
Since realizability games are memory-less determined, it is sufficient
to consider only local strategies.
Given a  strategy $\rhoB$ that is winning in $\GameB$ we define
a strategy $\rhoT$ in $\GameT$ as follows.
Assuming related positions, $\rhoT$ moves in $\GameT$ to the successor
that is related to the position where $\rhoB$ moves in $\GameB$.
By (3) above, it follows that for every play played in $\GameB$
according to $\rhoB$ there is a play in $\GameT$ played according to
$\rhoT$ that results in the same trace, and vice-versa: for every play
played in $\GameT$ according to $\rhoT$ there is a play in $\GameB$
played according to $\rhoB$ that results in the same trace.
Since $\rhoB$ is winning, so is $\rhoT$.
The other direction follows similarly, because again $\rhoB$ can be
constructed from $\rhoT$ not only guaranteeing the same valuation of
literals and corresponding variables, but also that the extra
requirement holds in the resulting position.
\qed
\end{proof}

%%%%%%%%%%%ORIGINAL theorem

\iffalse
%
\begin{theorem} 
  \label{temporalGlobalTheorem}
  $\phiT$ is realizable if and only if $\phiB$ is
  realizable.
\end{theorem}
%
\begin{proof}
  (\textit{Sketch}): The main element of the proof is that the system
  player wins the game for $\phiT$ if and only if the system player
  wins the game for $\phiB$.
  %
  This is proven by creating a binary relation $\mathcal{S}$ between
  the positions of the games showing (1) that the valuations of atomic
  propositions correspond to equivalent valuations of literals, and
  (2) that moves in one game can be simulated to corresponding moves
  in the other game (and vice-versa) leading to states related by
  $\mathcal{S}$.
  %
  This, in turn, induces a bijective map between winning strategies in
  one game and in the other, concluding that the system wins the game
  for $\phiT$ if and only if the system wins the game for $\phiB$.
  %
  \qed
\end{proof}
\fi

%%%%%%%%%%%END OF ORIGINAL theorem

The following corollary of Thm.~\ref{temporalGlobalTheorem} follows
immediately.
\begin{theorem} 
  \label{ltlDecidable}
  Let $\mathcal{T}$ be a theory with a decidable $\exists^*\forall^*$-fragment.
  Then, $\LTLt$ realizability is decidable.
\end{theorem}

%\begin{wrapfigure}[6]{l}{0.54\textwidth}
%  \begin{minipage}{0.55\textwidth}
%    %\vspace{-2.3em}
%    \begin{algorithm}[H]
%    \SetAlgorithmName{Alg.}{} 
%      
%      %$t \leftarrow 0$ and $E\leftarrow \emptyset$\;
%      \While{$t \in t$}{
%       %$E\leftarrow E\cup \semSym{\varphi}^t \cup \sem{A^t}_\varphi$\;
%        %\mbox{$E\leftarrow E\cup \{ x^t=v \;|\; \text{ for inputs $x$} \}$}\;
%        Evaluate and Simplify\;
%        Output\;
%        %Prune\;
%        %$t \leftarrow t + 1$ \;
%      }
%  %\caption{\mbox{Online Symbolic Monitor for $\varphi$}}
%  \caption{Online Monitor}
%  \label{alg:symonline}
%\end{algorithm}
%\end{minipage}
%%\vspace{-2em}
%\end{wrapfigure}

%%% Local Variables:
%%% TeX-master: "../main.tex"
%%% TeX-PDF-mode: t
%%% End:

%\newpage

\section{Efficient algorithms for Boolean Abstraction}
\label{sec:efficient}

\subsection{Quasi-reactions}

The basic algorithm presented in Section~\ref{sec:booleanAbs}
exhaustively traverses the set of reactions, one at a time, checking
whether each reaction is valid.
Therefore, the body of the loop is visited $2^{|\calC|}$ times.
In practice, the running time of this basic algorithm quickly becomes
unfeasible.

We now improve Alg.~\ref{algoBruteForce} by exploiting the observation
that every SMT query for the validity of a reaction reveals
information about the validity of other reactions.
We will exploit this idea by learning uninteresting subsequent sets of
reactions and pruning the search space.
The faster algorithms that we present below encode the remaining search
space using a SAT formula, whose models are further reactions to
explore.

To implement the learning-and-pruning idea we first introduce the
notion of quasi-reaction.

\begin{definition}[Quasi-reaction]
  A quasi-reaction is a pair $(P,A)$ where $P\subseteq{}\calC$,
  $A\subseteq{}\calC$ and  $P \cap A=\emptyset$.
\end{definition}

Quasi-reactions remove from reactions the constraint that
$P \cup A=\calC$.
A quasi-reaction represents the set of reactions that would be
obtained from choosing the remaining \configurations that are neither
in $P$ nor in $A$ as either potential or antipotential.
The set of quasi-reactions is:
\[
  \calQ =\{ (P,A) | P,A \subseteq \calC \textit{ and }  P \cap A=\emptyset \}
\]
Note that $\calR=\{(P,A)\in\calQ | P \cup A =\calC\}.$

\begin{example}
Consider a case with four \configurations $c_0$, $c_1$, $c_2$ and $c_3$.
The quasi-reaction $(\{c_0,c_2\},\{c_1\})$ corresponds to the
following formula:
\[
\exists \overline{x} \textit{. } \big(\exists \overline{y}
\textit{. }f(c_0(\overline{x}, \overline{y})) \wedge \forall \overline{y}
\textit{. } \neg f(c_1(\overline{x}, \overline{y})) \wedge \exists
\overline{y} \textit{. }f(c_2(\overline{x}, \overline{y}))\big)
\]
Note that nothing is stated in this quasi-reaction about $c_3$ (it
neither acts as a potential nor as an antipotential).
\qed
\end{example}

Consider the following order between quasi-reactions:
$(P,A)\qrprec(P',A')$ holds if and only if $P\subseteq P'$ and
$A\subseteq A'$.
It is easy to see that $\qrprec$ is a partial order, that
$(\emptyset,\emptyset)$ is the lowest element and that for every two
elements $(P,A)$ and $(P',A')$ there is a greatest lower bound (namely
$(P\cap P',A\cap A')$). Therefore
$(P,A)\sqcap(P',A')\DefinedAs(P\cap P',A\cap A')$ is a meet operation
(it is associative, commutative and idempotent).
Note that $q\qrprec{}q'$ if and only if $q\sqcap q'=q$.
Formally:

\begin{proposition}
  $(\calQ,\sqcap)$ is a lower semi-lattice.
\end{proposition}

The quasi-reaction semi-lattice represents how \textit{informative} a
quasi-reaction is.
Given a quasi-reaction $(P,A)$, removing an element from either $P$ or
$A$ results in a strictly less informative quasi-reaction.
The lowest element $(\emptyset,\emptyset)$ contains the least
information.

Given a quasi-reaction $q$, the set
$\calQ_q=\{ q'\in \calQ | q'\qrprec q\}$ of the quasi-reactions below
$q$ form a full lattice with join
$(P,Q)\sqcup(P',Q')\DefinedAs (P\cup P',Q\cup Q')$.
This is well defined because $P'$ and $Q$, and $P$ and $Q'$ are
guaranteed to be disjoint.

\begin{proposition}
  For every $q$, $(\calQ_q,\sqcap,\sqcup)$ is a lattice.
\end{proposition}

As for reactions, quasi-reactions correspond to a formula in the
theory as follows:
\[
  \qreact_{(P,A)}(\xs) = \bigwedge_{c \in P} \big(\exists \ys. c(\xs,\ys)\big) \wedge \bigwedge_{c \in A} \big(\forall \ys.\neg c(\xs,\ys)\big)
\]
Again, given a quasi-reaction $q$, if $\exists \xbar. \qreact_q(\xbar)$
is valid we say that $q$ is valid, otherwise we say that $q$ is
invalid.
The following holds directly from the definition (and the fact that
adding conjuncts makes a first-order formula ``less satisfiable'').
\begin{proposition}
  \label{prop:below}
  Let $q,q'$ be two quasi-reactions with $q\qrprec{}q'$. If $q$ is
  invalid then $q'$ is invalid. If $q'$ is valid then $q$ is valid.
\end{proposition}
These results enable the following optimizations.

\subsection{Quasi-reaction-based Optimizations}
\label{subsec:quasiOptimizations}

\subsubsection{A Logic-based Optimization.}
Consider that, during the search for valid reactions in the main loop,
a reaction $(P,A)$ is found to be invalid, that is $\react_{(P,A)}$ is
unsatisfiable.
If the algorithms explores the quasi-reactions below $(P,A)$, finding
$(P',A')\qrprec(P,A)$ such that $\qreact_{(P',A')}$, then by
Prop.~\ref{prop:below}, every reaction $(P'',A'')$ above $(P',A')$ is
guaranteed to be invalid.
This allows to prune the search in the main loop by computing a more
informative quasi-reaction $q$ after an invalid reaction $r$ is found,
and skipping all reactions above $q$ (and not only $r$).
For example, if the reaction corresponding to
$(\{c_0,c_2,c_3\},\{c_1\})$ is found to be invalid, and by exploring
quasi-reactions below it, we find that $(\{c_0\},\{c_1\})$ is also
invalid, then we can skip all reactions above $(\{c_0\},\{c_1\})$.
This includes for example $(\{c_0,c_2\},\{c_1,c_3\})$ and
$(\{c_0,c_3\},\{c_1,c_2\})$.
In general, the lower the invalid quasi-reaction in $\qrprec$, the more
reactions will be pruned.
This optimization resembles a standard choosing of max/min elements in
an anti-chain.

\subsubsection{A Game-based Optimization.}
Consider now two reactions $r=(P,A)$ and $r'=(P',A')$ such that
$P\subseteq P'$ and assume that both are valid reactions.
Since $r'$ allows more \configurations to the system (because the
potentials $P$ determine these choices), the environment player will
always prefer to play $r$ than $r'$.
Formally, if there is a winning strategy for the environment that
chooses values for $\xbar$ (corresponding to a model of $\react_r$),
then choosing values for $\xbar'$ instead (corresponding to a model of
$\react_{r'}$) will also be winning.

Therefore, if a reaction $r$ is found to be valid, we can prune the
search for reactions $r'$ that contain strictly more potentials,
because even if $r'$ is also valid, it will be less interesting for
the environment player.
For instance, if $(\{c_0,c_3\},\{c_1,c_2\})$ is valid, then
$(\{c_0,c_1,c_3\},\{c_2\})$ and $(\{c_0,c_1,c_3,c_2\},\{\})$ become
uninteresting to be explored and can be pruned from the search.

\subsection{A Single Model-loop Algorithm (Alg.~\ref{algoModelLoop})}

We present now a faster algorithm that replaces the main loop of
Alg.~\ref{algoBruteForce} that performs exhaustive exploration with a
SAT-based search procedure that prunes uninteresting reactions.
In order to do so, we use a SAT formula $\psi$ with one variable $z_i$
per \configuration $c_i$, in a DPLL(T) fashion.
An assignment $v:\Vars(\psi)\Into\Bool$ to these variables represents
a reaction $(P,A)$ where
\[
  P=\{ c_i |  v(z_i)=\True\} \hspace{3em}  A=\{ c_j |  v(z_j)=\False \} 
\]
Similarly, a partial assignment $v:\Vars(\psi)\Part\Bool$ represents a
quasi-reaction.
The intended meaning of $\psi$ is that its models encode the set of
interesting reactions that remain to be explored.
This formula is initialized with $\psi=\True$ (note that
$\neg (\bigwedge_{z_i} \neg z_i)$ is also a correct starting point
because the reaction where all \configurations
are antipotentials is invalid).
Then, a SAT query is used to find a satisfying assignment for $\psi$,
which corresponds to a (quasi-)reaction $r$ whose validity is
\begin{wrapfigure}[25]{l}{0.46\textwidth}
  \begin{minipage}{0.46\textwidth}
    \vspace{-2.3em}
\begin{algorithm}[H]
  \AlgModelLoop
  \caption{Model-loop}% based Boolean Abstraction algorithm}
  \label{algoModelLoop}
\end{algorithm}
\end{minipage}
%\vspace{-2em}
\end{wrapfigure}
interesting to be explored.
Alg.~\ref{algoModelLoop} shows the Model-loop algorithm.
The three
main building blocks of the model-loop algorithm are:
%\begin{enumerate}[(1)]
\begin{compactenum}[(1)]
\item Alg.~\ref{algoModelLoop} stops when $\psi$ is invalid (line $14$).
\item To explore a new reaction, Alg.~\ref{algoModelLoop} obtains a
  satisfying assignment for $\psi$ (line $15$).
\item Alg.~\ref{algoModelLoop} checks the validity of the reaction
  (line $16$) and enriches $\psi$ o prune according to what can be
  learned, as follows:
\begin{itemize}
%\begin{compactitem}
\item If the reaction is invalid (as a result of the SMT query in line
  $16$), then it checks the validity of quasi-reaction
  $q=(\emptyset,A)$ in line $23$.  If $q$ is invalid, add the negation
  of $q$ as a new conjunction of $\psi$ (line $26$).  If $q$ is valid,
  add the negation of the reaction (line $24$).
  This prevents all SAT models that agree with one of these $q$, which
  correspond to reactions $q\qrprec{}r'$, including $r$.
\end{itemize}

\begin{itemize}
\item If the reaction is valid, then it is added to the set of valid
  reactions \VR and the corresponding quasi-reaction that results from
  removing the antipotentials is added (negated) to $\psi$ (line
  $18$), preventing the exploration of uninteresting cases, according
  to the game-based optimization.
  % \end{compactitem}
  \end{itemize}
 \end{compactenum}
%\end{enumerate}

As for the notation in Alg.~\ref{algoModelLoop} (also in
Alg.~\ref{algoNestedSAT} and Alg.~\ref{algoInnerLoop}),
\textit{model($\psi$)} in line $15$ is a function that returns a
satisfying assignment of the SAT formula $\psi$,
\textit{posVars(m)} returns the positive variables of $m$ (e.g.,
$c_i, c_j$ etc.) and \textit{negVars(m)} returns the negative
variables. 
Finally, $ \textit{toTheory}(m, \mathcal{C}) = \bigwedge_{m_i} c_i^p \wedge \bigwedge_{\neg m_i} c_i^a$
(in lines $16$ and $23$) translates a Boolean formula into its
corresponding formula in the given $\mathcal{T}$ theory.
Note that unsatisfiable $m$ can be minimized finding cores. % \cite{bendikMeel2021countingMaximalUnsatisfiableS}.

If $r$ is invalid and $(\emptyset,A)$ is found also to be
invalid, then exponentially many cases can be pruned.
Similarly, if $r$ is valid, also exponentially many cases can be pruned.
The following result shows the correctness of Alg.~\ref{algoModelLoop}:

\begin{theorem}
  Alg.~\ref{algoModelLoop} terminates and outputs a correct Boolean
  abstraction.
\end{theorem}

\begin{proof}(Sketch). Alg.~\ref{algoModelLoop} terminates because, at
  each step in the loop, $\psi$ removes at least one satisfying
  assignment and the total number is bounded by $2^{|\calC|}$.
  Also, the correctness of the generated formula is guaranteed
  because, for every valid reaction in Alg.~\ref{algoBruteForce},
  either there is a valid reaction found in Alg.~\ref{algoModelLoop}
  or a more promising reaction found in Alg.~\ref{algoModelLoop}.
  \qed
\end{proof}

\subsection{A Nested-SAT algorithm (Alg.~\ref{algoNestedSAT})}
%
%\label{subsubsec:multipleSAT}

We now present an improvement of Alg.~\ref{algoModelLoop} that
performs a more detailed search for a promising collection of invalid
quasi-reactions under an invalid reaction $r$.
\begin{wrapfigure}[21]{l}{0.46\textwidth}
  \begin{minipage}{0.46\textwidth}
    \vspace{-2.3em}
\begin{algorithm}[H]
  \AlgNestedSAT
  \caption{Nested-SAT}% based Boolean Abstraction algorithm}
  \label{algoNestedSAT}
\end{algorithm}
\end{minipage}
%\vspace{-2em}
\end{wrapfigure}
Note that it is not necessary to find the precise collection of all
the smallest quasi-reactions that are under an invalid reaction $r$,
as long as at least
one quasi-reaction under $r$ is calculated
(perhaps, $r$ itself).
Finding lower quasi-reactions allow to prune more, but its calculation
is more costly, because more SMT queries need to be performed.
The Nested-SAT algorithm (Alg.~\ref{algoNestedSAT}) explores (using an
inner SAT encoding) this trade-off between computing more exhaustively
better invalid quasi-reactions and the cost of the search.
The three
main building blocks of the nested-SAT algorithm (see
Alg.~\ref{algoNestedSAT}) are:
\begin{itemize}
%\begin{compactenum}[(1)]
\item[(1)] It stops when $\psi$ is invalid (as in
  Alg.~\ref{algoModelLoop}), in line $33$.
\item[(2)] To get the reaction, obtain a satisfying assignment $m$ for
  $\psi$ (as in Alg.~\ref{algoModelLoop}), in line $34$.
\end{itemize}

\begin{itemize}
\item[(3)] Check the validity of the corresponding reaction and prune
  $\psi$ according to what can be learned as follows.
  If the reaction is valid, then we proceed as in
  Alg.~\ref{algoModelLoop}.
  If $r=(P,A)$ is invalid (as a result of the SMT query), then an
  inner SAT formula encodes whether a \configuration is masked
  (eliminated from $P$ or $A$).
  Models of the inner SAT formula, therefore, correspond to
  quasi-reactions below $r$.
  If a quasi-reaction $q$ found in the inner loop is invalid, the inner
  formula is additionally constrained and the set of invalid
  quasi-reactions is expanded.
  If a quasi-reaction $q$ found is valid, then the inner SAT formula
  is pruned eliminating all quasi-reactions that are guaranteed to be
  valid.
  At the end of the inner loop, a (non-empty) collection of invalid
  quasi-reactions are added to $\psi$.
\end{itemize}

The inner loop, shown in Alg.~\ref{algoInnerLoop} (where \textit{VQ}
stands for \textit{valid quasi-reactions}),
\begin{wrapfigure}[17]{l}{0.45\textwidth}
  \begin{minipage}{0.45\textwidth}
    \vspace{-2.3em}
\begin{algorithm} [H]
  \AlgInnerLoop
 \caption{Inner loop}
 \label{algoInnerLoop}
\end{algorithm}
\end{minipage}
%\vspace{-2em}
\end{wrapfigure}
explores a full lattice.
Also, note that $\neg (\bigwedge_{z_i} \neg z_i)$ is, again, a
correct starting point.
Consider, for example, that the outer loop finds
$(\{c_1,c_3\},\{c_0,c_2\})$ to be invalid and that the inner loop
produces assignment $w_0\And w_1\And w_2 \And \neg w_3$.
This corresponds to $c_3$ being masked producing quasi-reaction
$(\{c_1\},\{c_0,c_2\})$.
The pruning system is the following:

%\begin{itemize}
\begin{compactitem}
\item If quasi-reaction $q$ is valid then the inner SAT formula is
  pruned eliminating all inner models that agree with the model in the
  masked \configurations.
  In our example, we would prune all models that satisfy $\neg w_3$ if
  $q$ is valid (because the resulting quasi-reactions will be
  inevitably valid).
\end{compactitem}

\begin{compactitem}
\item If quasi-reaction $q$ is invalid, then we prune in the inner
  search all quasi-reactions that mask less than $q$, because these
  will be inevitably invalid.
  In our example, we would prune all models satisfying
  $\neg (w_0\And w_1\And w_2)$.
 \end{compactitem}
%\end{itemize}
% 
% Note that inner loops' searches are performed in a (full) lattice.
%
%
Note that %$ \textit{toTheory}(u, m, \mathcal{C}) = \bigwedge_{m_i} c_i^p \wedge \bigwedge_{\neg m_i} c_i^a$
$\textit{toTheory\_inn}(u, m, \mathcal{C}) = \bigwedge_{m_i \wedge u_j} c_i^p \wedge \bigwedge_{\neg m_i \wedge u_j} c_i^a$
is not the same function as the $\textit{toTheory()}$ used in Alg.~\ref{algoModelLoop} and Alg.~\ref{algoNestedSAT}, 
since the inner loops needs both model $m$ and mask $u$ (which makes no sense to be negated)
to translate a Boolean formula into a $\mathcal{T}$-formula.
Also, note that there is again a trade-off in the inner loop because an
exhaustive search is not necessary.
Thus, in practice, we also used some basic heuristics: (1) entering the
inner loop only when $(\emptyset,A)$ is invalid; (2) fixing a maximum
number of inner model queries per outer model with the possibility to decrement this amount
dynamically with a decay; and (3) reducing the number of times the
inner loop is exercised (e.g., \textit{enter the inner loop only if the number
of invalid outer models so far is even}).

\begin{example}
  \label{ex:alg3}
  We explore the results of Alg.~\ref{algoNestedSAT}.
  A possible execution for 2 literals can be as follows:
\begin{enumerate}
%  \begin{compactenum}
\item Reaction $(\{c_0,c_3\},\{c_1,c_2\} )$ is obtained in line $34$,
  which is declared invalid by the SMT solver in line $35$.
  The inner loop called in line $42$ produces $(\{c_0\},\{c_1\} )$,
  $(\{c_3\},\{c_2\} )$ and $(\{\},\{c_1,c_2\} )$ as three invalid
  quasi-reactions, and their negations are added to the SAT formula of
  the outer loop in line $43$.
\item A second reaction $(\{c_0,c_1\},\{c_3,c_4\} )$ is obtained from
  the SAT solver in line $34$, and now the SMT solver query is valid
  in line $35$.
  Then, $\neg(c_0\And c_1)$ is added to the outer SAT formula in line $37$.
\item A third reaction $(\{c_2,c_3\},\{c_0,c_1\} )$ is obtained in
  line $33$ , which is again valid in line $35$.
  Similarly, $\neg(c_2\And c_3)$ is added the outer SAT formula in
  line $37$.
\item A fourth reaction $(\{c_1,c_2\},\{c_0,c_3\} )$ is obtained in
  line $33$, which is now invalid (line $35$).
  The inner loop called in line $42$ generates the following cores:
  $(\{c_1\},\{c_0\} )$ and $(\{c_2\},\{c_3\})$.
  The addition of the negation of these cores leads to an
  unsatisfiable outer SAT formula, and the algorithm terminates.
  \end{enumerate}
  
The execution in this example has performed 4 SAT+SMT queries in the
outer loop, and 3+2 SAT+SMT queries in the inner loops.
The brute-force Alg.~\ref{algoBruteForce} would have performed 16 queries.
Note that the difference between the exhaustive version and the optimisations 
soon increases exponentially when we consider specifications with more literals.

\qed
\end{example}

%%% Local Variables:
%%% TeX-master: "../main.tex"
%%% TeX-PDF-mode: t
%%% End:

\section{Empirical evaluation}
\label{sec:empirical}

%\subsubsection{Evaluation of the tool.}
%
We perform an empirical evaluation on six specifications 
inspired by real industrial cases: \textit{Lift} (\textit{Li.}), \textit{Train}
(\textit{Tr.}), \textit{Connect} (\textit{Con.}), \textit{Cooker}
(\textit{Coo.}), \textit{Usb} (\textit{Usb}) and \textit{Stage}
(\textit{St.}), and a
synthetic %\footnote{Please, do not confuse with \textit{synthesis}.}
example (\textit{Syn.}) with versions from 2 to 7 literals.
For the implementation, we used used Python $3.8.8$ with Z3 $4.11$.

\TableBenchmark

It is easy to see that ``clusters'' of literals that do not share
variables can be Booleanized independently, so we split into clusters
each of the examples.
We report our results in Fig.~\ref{tab:benchmark}.
Each row contains the result for a cluster of an experiment (each one
for the fastest heuristic).
Each benchmark is split into clusters, where we show the number of
variables (\textit{vr}.) and literals (\textit{lt.}) per cluster.
We also show running times of each algorithm against each cluster;
concretely, we test Alg.~\ref{algoBruteForce} (\textit{BF}), Alg.~\ref{algoModelLoop} (\textit{SAT}) and
Alg.~\ref{algoNestedSAT} (\textit{Doub}.).
For Alg.~\ref{algoModelLoop} and Alg.~\ref{algoNestedSAT}, we show the number of queries performed; in the
case of Alg.~\ref{algoNestedSAT}, we also show both outer and inner queries.
Alg.~\ref{algoBruteForce} and Alg.~\ref{algoModelLoop} require no heuristics.
For Alg.~\ref{algoNestedSAT}, we report, left to right: maximum number
of inner loops (\textit{MxI.}), the modulo division criteria
(\textit{Md.})\footnote{This means that the inner loop is entered if and
  only if the number of invalid models so far is divisible by
  \textit{Md}, and we found \textit{Md} values of $2$, $3$ and $20$ to be interesting.}, the number of queries after which we perform a decay
of $1$ in the maximum number of inner loops (\textit{Dc.}), and if we
apply the invalidity of $(\emptyset, A)$ as a criteria to enter the
inner loop ($A.$), where $\checkmark$ means that we do and $\times$
means the contrary.
 Also,  $\perp$ means timeout (or \textit{no data}).

 The brute-force (BF) Alg.~\ref{algoBruteForce} performs well with 3
 or fewer literals, but the performance dramatically decreases with 4
 literals.
 Alg.~\ref{algoModelLoop} (single SAT) performs well up to 4 literals,
 and it can hardly handle cases with 6 or more literals.
An exception is \textit{Lift (1,7)} which is simpler since it has only
one variable (and this implies that there is only one player).
The performance improvement of SAT with respect to BF is due to the
decreasing of queries.
For example, \textit{Train (3,6)} performs $13706$ queries, whereas BF
would need $2^{2^6}=1.844\cdot10^{18}$ queries.

All examples are Booleanizable when using Alg.~\ref{algoNestedSAT} (two SAT loops),
particularly when using a combination of concrete heuristics.
For instance, in small cases (2 to 5 literals) it seems that
heuristic-setups like $3/3/3/0/\checkmark$\footnote{This means: we
  only perform 3 inner loop queries per outer loop query (and there is
  no decay, i.e., $decay=0$), we enter the inner loop once per 3 outer
  loops and we only enter the inner loop if $(\emptyset, A)$ is
  invalid.} are fast, whereas in bigger cases other setups like
$40/2/0/\checkmark$ or $100/40/20/\times$ are faster.
We conjecture that a non-zero decay is required to handle large
inputs, since inner loop exploration becomes less useful after some
time.
However, adding a decay is not always faster than fixing a number of
inner loops (see \textit{Syn (2,7)}), but it always yields better
results in balancing the number of queries between the two nested SAT
layers.
Thus, since balancing the number of queries typically leads to faster
execution times, we recommend to use decays.
Note that we performed all the experiments reported in this section running all
cases several times and computing averages, because Z3 exhibited a big
volatility in the models it produces, which in turn influenced the
running time of our algorithms.
This significantly affects the precise reproducibility of the running times.
For instance, for \textit{Syn(2,5)} the worst case execution was
almost three times worst than the average execution reported in
Fig.~\ref{tab:benchmark}.
Studying this phenomena more closely is work in progress.
Note that there are cases in which the number of queries of \textit{SAT} and \textit{Doub}. are the same 
(e.g., \textit{Usb(3,5)}), 
which happened when the \textit{A.} heuristic had the effect of making the search not to enter the inner loop.
%
%The reader will also note that whenever this happens, it also happens that the number of inner loops in Doubs is 0 and that heuristic A is also activated. 
%This suggests that there have been no cases in Doub where the inner loop has been accessed, because A has not been fulfilled.

In Fig.~\ref{tab:benchmark} we also analyzed the constructed $\phiB$, 
measuring the number of valid reactions from which it is made (\textit{Val.}) 
and the time (\textit{Tme.}) that a realizability checker  takes to verify whether 
$\phiB$ (hence, $\phiT$) is realizable or not (expressed with dark and light gray colours, respectively).
We used Strix \cite{meyer18strix} as the
realizability checker.
As we can see, there is a correspondence between the expected realizability in $\phiT$ and the
realizability result that Strix returns in $\phiB$.
Indeed, we can see all instances can be solved in less than $7$
seconds, and the length of the Boolean formula (characterized by the
number of valid reactions) hardly affects performance.
This suggests that future work should be focused on reducing time
necessary to produce Boolean abstraction to scale even further.

Also, note that Fig.~\ref{tab:benchmark} shows remarkable results as
for ratios of queries required with respect to the (doubly
exponential) brute-force algorithm: e.g., $4792+9941$ (outer + inner
loops) out of the $1.844 \cdot 10^{19}$ queries that the brute-force
algorithm would need, which is less than its $1 \cdot 10^{-13}\%$ (see
Fig.~\ref{tab:coverage} for more details).
We also compared the performance
and number of queries for two different theories
$\mathcal{T}_{\mathbb{Z}}$ and $\mathcal{T}_{\mathbb{R}}$
for \textit{Syn (2,3)} to \textit{Syn (2,6)}.
Note, again, that the realizability result may vary if a specification is
interpreted in different theories, but this is not relevant for the experiment in Fig.~\ref{tab:theoryComparison},
which suggests that time results are not dominated by the SMT solver; 
but, again, from the enclosing abstraction algorithms.
\begin{figure}[t]
  \centering
%\begin{center}
\begin{tabular}{|c | c | c | c | c |} 
 \hline
 Lits & Alg. & Performed queries (out+inn) & Out of & Needed queries ($\simeq\%$) \\ 
 %(2 to 12) & (outer+inner) & (exponential) & ($\simeq\%$) \\ [0.5ex] 
 \hline\hline
 $2$ & Alg 2 & $4$ & $16$ & $25$ \\ 
 \hline
 $3$ & Alg 2 & $8$ & $256$ & $3.125$ \\
 \hline
 $4$ & Alg 3 & $83+380$ & $65536$ & $0.709$ \\
 \hline
 $5$ & Alg 3 & $380+2800$ & $4294967296$ & $7.404 \cdot 10^{-5}$ \\
 \hline
 $6$ & Alg 3 & $4792+9941$ & $1.844 \cdot 10^{19}$ & $1 \cdot 10^{-13}$ \\
 \hline
 %$7$ & Alg 3 & $8363+125940$ & $3.402 \cdot 10^{38}$ & $3.947 \cdot 10^{-32}$ \\
 %\hline
 ... & ... & ... &  ... & ... \\
 \hline
 $12$ & Alg 3 & $2728+40920$ &  $\infty$ & $0$ \\
 \hline
\end{tabular}
% \caption{This table of coverage shows which percentage of the double
%   exponential has been needed to visit in the best algorithm for each
%   input size We can see that, with large inputs, the saved queries
%   make Boolean abstraction usable in practise.}
\caption{Best numbers of queries for Alg.~2 and 3 relative to
  brute-force (Alg.1).}
\label{tab:coverage}
%\end{center}
\end{figure}
%
%
%RESULTS FROM: https://drive.google.com/drive/folders/17Js4TL9F1Ci4HH5qtWlH6FkDfkypbslW
\begin{figure}[t]
  \centering
%\begin{center}
\begin{tabular}{|c|c|c c|c c|} 
 \hline
 \multirow{2}{*}{Lits} & \multirow{2}{*}{Heuristic} & \multicolumn{2}{|c|}{$\mathcal{T}_{\mathbb{Z}}$} & \multicolumn{2}{c|}{$\mathcal{T}_{\mathbb{R}}$} \\
 & setup & Time (s) & Queries (ou/in) & Time (s) & Queries (ou/in) \\
 %Literals & Performed queries & Out of & Needed queries \\ 
 %(2 to 12) & (outer+inner) & (exponential) & ($\simeq\%$) \\ [0.5ex] 
 \hline\hline
 $3$ & $10/2/0/\checkmark$ & $0.63$ & $8 / 30$ & $0.90$ & $14 / 40$ \\ 
 \hline
 $4$ & $10/2/0/\checkmark$ &  $16.14$ & $308 / 500$ & $11.19$ & $125 / 560$ \\
 \hline
  $5$ & $20/2/0/\checkmark$ &  $62.44$ & $408 / 3220$ & $88.55$ & $357 / 3460$ \\
 \hline
  $6$ & $40/2/0/\checkmark$ &  $678.71$ & $2094 / 32760$ & $722.64$ & $1862 / 35840$ \\
 \hline
\end{tabular}
\caption{Comparison of $\ThZ$ and $\ThR$ for \textit{Syn (2,3)} to \textit{Syn (2,6).}}
% \caption{A preliminar comparison of double-SAT with \textit{Synt} 3 to
%   6 interpreted in $\mathcal{T}_{\mathbb{Z}}$ and interpreted in
%   $\mathcal{T}_{\mathbb{R}}$. We observe results are similar in
%   $\mathcal{T}_{\mathbb{Z}}$ and $\mathcal{T}_{\mathbb{R}}$.}
\label{tab:theoryComparison}
%\end{center}
\end{figure}

\section{Related Work and Conclusions} 

\subsubsection{Related work.}
Constraint LTL
\cite{demriDSouza2002automataTheoreticApproachConstraintLTL} extends
LTL with the possibility of expressing constraints between variables
at bounded distance (of time).
The theories considered are a restricted form of $\ThZ$ with only
comparisons with additional restrictions to overcome undecidability.
In comparison, we do not allow predicates to compare variables at
different timesteps, but we prove decidability for all theories with an
$\exists^*\forall^*$ decidable fragment.
LTL modulo theories is studied in
\cite{gianolaETAL2022LTLmoduloTheoriesOverFiniteT,faranKupferman2022LTLwithArithmeticApplicationsR} for finite traces
and they allow temporal operators within predicates, leading the logic to
undecidability.

As for works closest to ours,~\cite{cheng2013numerical} proposes
numerical LTL synthesis using an interplay between an LTL synthesizer
and a non-linear real arithmetic checker.
However,~\cite{cheng2013numerical} overapproximates the power of the
system and hence it is not precise for realizability.
Linear arithmetic games are studied in
\cite{azadehKincaid2017strategySynthesisLinearArithmetic} 
introducing algorithms for synthesizing winning strategies for
non-reactive specifications.
Also, \cite{katisETAL2016synthesisAssumeGuaranteeContracts} considers infinite theories (like us), 
but it does not guarantee success or termination, whereas our Boolean abstraction is complete. 
They only consider safety, while our approach considers all LTL.
The follow-up \cite{katisETAL2018validityGuidedSynthesis} has still similar limitations: 
only liveness properties that can be reduced to safety are accepted, and guarantees termination only for the unrealizability case.
Similarly, \cite{gacekETAL2015towardsRealizabilityCheckingContracts} is incomplete, 
and requires a powerful solver for many quantifier alternations, which can be reduced to 1-alternation, 
but at the expense of the algorithm being no longer sound for the unrealizable case (e.g., depends on Z3 not answering ``unknown'').
As for \cite{walkerRyzhyk2014predicateAbstractionReactiveS}, it  (1) only considers safety/liveness GR(1) specifications, 
(2) is limited to the theory of fixed-size vectors and requires (3) quantifier elimination (4) and guidance. 
We only require $\exists^*\forall^*$-satisfiability (for Boolean abstraction) and we consider multiple infinite theories. 
The usual main difference is that Boolean abstraction generates a (Boolean) LTL specification so that existing tools 
can be used with any of their internal techniques and algorithms (bounded synthesis, for example) 
and will automatically benefit from further optimizations. 
Moreover, it preserves fragments like safety and GR(1) so specialized solvers can be used. 
On the contrary, all approaches above adapt one specific technique and implement it in a monolithic way. 

Temporal Stream Logic
(TSL)~\cite{finkbeinerETAL2017temporalStreamLogicSynthesisBeyondBools}
extends LTL with complex data that can be related accross time, making
use of a new \textit{update} operator
$\llbracket y \mapsfrom f x \rrbracket$, to indicate that $y$ receives
the result of applying function $f$ to variable $x$.
TSL is later extended to theories in
\cite{finkbeinerETAL2021temporalStreamLogicModuloTheories,maderbacherBloem2021reactiveSynthesisModuloTheoriesAbstraction}.
In all these works, realizability is undecidable.
Also, in \cite{wonhyukETAL2022canSynthesisSyntaxBeFriends} reactive
synthesis and syntax guided synthesis
(SyGuS)~\cite{rajeevETAL2013syntaxGuidedSynthesis} collaborate in the
synthesis process, and generate executable code that guarantees
reactive and data-level properties.
It also suffers from undecidability: both due to the undecidability of
TSL \cite{finkbeinerETAL2017temporalStreamLogicSynthesisBeyondBools}
and of SyGus~\cite{caulfieldETAL2015whatsDecidableAboutSyntaxGuidedS}.
In comparison, we cannot relate values accross time but we
provide a decidable realizability procedure.

Comparing TSL with \LTLt, TSL is undecidable already for safety, the
theory of equality and Presburger arithmetic. More precisely, TSL is
only known to be decidable for three fragments (see Thm. 7 in
\cite{finkbeinerETAL2021temporalStreamLogicModuloTheories}).
TSL is (1) semi-decidable for the reachability fragment of TSL (i.e.,
the fragment of TSL that only permits the next operator and the
eventually operator as temporal operators); (2) decidable for formulae
consisting of only logical operators, predicates, updates, next
operators, and at most one top-level eventually operator; and (3)
semi-decidable for formulae with one cell (i.e., controllable
outputs).
All the specifications considered for empirical evaluation in
Section~\ref{sec:empirical} are not within the considered decidable or
semi-decidable fragments.
Also, TSL allows (finite) uninterpreted predicates, whereas we need to
have predicates well defined within the semantics of theories of
specifications for which we perform Boolean abstraction.

\subsubsection{Conclusion.}
The main contribution of this paper is to show that \LTLt is
decidable via a Boolean abstraction technique for all theories of data
with a decidable $\exists^*\forall^*$ fragment.
Our algorithms create, from a given \LTLt specification where atomic
propositions are literals in such a theory, an equi-realizable
specification with Boolean atomic propositions.
We also have introduced efficient algorithms using SAT solvers for
efficiently traversing the search space.
A SAT formula encodes the space of reactions to be explore and our
algorithms reduce this space by learning uninteresting areas from each
reaction explores.
The fastest algorithm uses a two layer SAT nested encoding, in a DPLL(T) fashion.
This search yields dramatically more efficient running times and makes
Boolean abstraction applicable to larger cases.
We have performed an empirical evaluation of implementations of our
algorithms.
We found empirically that the best performances are obtained when
there is a balance in the number of queries made by each layer of the
SAT-search.
To the best of our knowledge, this is the first method to propose a
solution (and efficient) to realizability for general
$\exists^*\forall^*$ decidable theories, which include, for instance,
the theories of integers and reals.

Future work includes first how to improve scalability further.
We plan to leverage quantifier elimination procedures \cite{cooper1972theoremProving}
%(e.g., ~\cite{cooper1972theoremProving}) 
to produce candidates for the sets
of valid reactions and then check (and correct) with faster algorithms.
Also, optimizations based in quasi-reactions can be enhanced if
state-of-the-art tools for satisfiability core search 
(e.g., \cite{liffitonETAL2016fastFlexibleMUSEnum,bendikMeel2021countingMinimalUnsatisfiableS,bendikMeel2021countingMaximalUnsatisfiableS}) are used.
Another direction is to extend our realizability method into a
synthesis procedure by synthesizing functions in
$\calT$ to produces witness values of variables controlled by
the system given (1) environment and system moves in the Boolean game,
and (2) environment values (consistent with the environment move).
Finally, we plan to study how to extend \LTLt with controlled transfer
of data accross time preserving decidability. %of synthesis.

\vfill

\pagebreak

%%% Local Variables:
%%% TeX-master: "../main.tex"
%%% TeX-PDF-mode: t
%%% End:

\bibliographystyle{plain}
\bibliography{references}

\vfill
\pagebreak
\appendix

\section{More about empirical evaluation} \label{appSec:moreEmpirical}

%We include here a number of questions that could not appear in the paper for space reasons.

%//
%Note that the original synthetic examples are expressed in standard (future) LTL, 
%whereas the industrial ones are in Past-LTL \cite{gabbayETAL1980onTemporalAnalysisFairness,Gabbay1987DeclarativePast}
%and use the \textit{Zyesterday} operation \cite{tonetta2017linearTimeEventFreezingF}
%that is true in the initial state. 
%
%Rather than using two different specialised realizability checkers for each temporality form, 
%we consider it a better option to use one tool to maintain consistency in the results table. 
%Thus, due to other ongoing research and that PLTL is more succint than LTL \cite{markey2003temporalLogicPastExponentiallySuccint} 
%(in addition to the fact that industrial examples were the majority), 
%we opted to use the PLTL semantics, so we translated the synthetic examples
%to past (using LTL2PLTL) and  used the AbsSynthe (safety) checker \cite{brenguierGuillermo2014AbsSynthe},
%which obtained outstanding results in the Reactive Synthesis Competition editions it participated 
%\footnote{See \url{http://www.syntcomp.org/} from 2014 to 2020, both included.}.
%We constructed our benchmark for realizability in the standard AIGER \footnote{See at \url{http://fmv.jku.at/aiger/} the main page.} 
%format \cite{biere2007AIGERandInvertedGraph} using 
%Py-Aiger PLTL \footnote{See \url{https://github.com/mvcisback/py-aiger-past-ltl} for more information.}.

In this paper, we only optimized heuristics (of Alg.~3) with respect to time and,
even if current evidence suggests that the number of valid reactions is not relevant, 
it could be the case it is relevant for other kind of formulae to be evaluated.
Thus, note that different heuristics yield different $\phiB$ that can be more succint 
(e.g., produce much less valid reactions): for instance, 
using the $100/20/40/\times$ heuristic-setup for \textit{Train(4,5)} took $4144$ seconds and produced $24$ valid reactions; 
whereas using the $20/2/0/\checkmark$ setup took $6328$ seconds, but produced $17$ valid reactions.
This means that the difference between using a set of heuristics 
or another one is not only performance of Boolean abstraction method, the difference can be as great 
as the (actual) possibility of performing realizability checking. 
Studying this phenomena is also reported as a work in progress: it might be the case that it is overall faster to spend more time 
on the $\LTLt$ to LTL encoding if this results in a formula with fewer valid reactions 
whose realizability result can be obtained faster.

Fig.~\ref{tab:coverage} contains the best ratios of queries required with
respect to the (doubly exponential) brute-force algorithm.
In Fig.~\ref{tab:theoryComparison}, we also compare the performance
and number of queries for \textit{Syn (2,3)} to \textit{Syn (2,6)} for theories
$\mathcal{T}_{\mathbb{Z}}$ and $\mathcal{T}_{\mathbb{R}}$.
%
%We can see that the performance is very similar, which suggests
%that---at least for arithmetic theories---the running time of the
%Boolean abstractions is not dominated by the SMT solver, but from the
%enclosing abstraction algorithms. 
Note, again, that the realizability result may vary if a specification is
interpreted in different theories, but this is not relevant for the experiment in Fig.~\ref{tab:theoryComparison}.

Also, we tried to replicate results of Fig.~2 using the AbsSynthe (safety) checker \cite{brenguierGuillermo2014AbsSynthe},
because of the fact that Past LTL~\cite{Gabbay1980OnTT} is more succint than LTL \cite{markey2003temporalLogicPastExponentiallySuccint} and
because it obtained remarkable results in the Reactive Synthesis Competition editions it participated 
\footnote{See \url{http://www.syntcomp.org/} from 2014 to 2020, both included.}.
Thus, we adapted some of our benchmark for realizability to the AIGER standard \footnote{See at \url{http://fmv.jku.at/aiger/} the main page.} 
format \cite{biere2007AIGERandInvertedGraph} using 
Py-Aiger PLTL \footnote{See \url{https://github.com/mvcisback/py-aiger-past-ltl} for more information.}.
However, even if early results in AbsSynthe are promising 
(e.g., it lasted $0.06$ seconds to solve \textit{Syn(2,2)} instead of the $4.12$ seconds of Strix 
\footnote{Note that we used a virtual machine for Strix, whereas we executed AbsSynthe locally.}), 
the parsing process is too expensive and we got timeouts for $\phiB$ 
containing many valid reactions such as \textit{Stages(8,8)}. 
Comparing generated $\phiB$  with different realizability checkers 
(including AbsSynthe)  is another interesting future research line.
\section{Benchmarks' literals}

%FIRST WRITTEN IN (DOC): https://docs.google.com/document/d/1fOq0i-qtSiJN019AvRRF3Kpntb_oye0ZZUj7J3E8BOg/edit

We show literals that compose each cluster, 
together with a minimal description about the specification from which they have been extracted.
Note that original names of variables and the rest of the specification 
(which includes state enumerated variables and Boolean variables) are not shown.%, since they would reveal industrial information.
Nevertheless, they are all safety specifications.

\subsection{Industrial case 1: Lift}

Lift is part of a set of specifications that describes the functioning of a freight elevator system.
%
%Concretely, different floors are selected on which to stop based on different weight parameters.
%
It is divided into 4 clusters, interpreted in $\mathcal{T}_{\mathbb{R}}$.

\subsubsection{Cluster 1: \textit{Lift (1,7)}.} This cluster contains 1 variable (which belongs to the system) and 7 literals. Concretely:
\[
\begin{array}{l@{\hspace{3em}}l}
  lit_1 = (v_0 = c_{21}) \\
  lit_4 = (v_0 = c_{19}) \\
  lit_5 = (v_0 = c_{20}) \\
  lit_7 = (v_0 = c_{22}) \\
  lit_{11} = (v_0 = c_{17}) \\
  lit_{12} = (v_0 = c_{18}) \\
  lit_{25} = \bigwedge(v_0 \geq 0, v_0 \leq 4)
\end{array}
\]

Note that the $c$ are just predefined constants and their value may affect the realizability result.
Also, note that we can have formulae within literals (e.g., $lit_{25}$).

\subsubsection{Cluster 2: \textit{Lift (2,4)}.} This cluster contains 2 variables (both of them belong to the environment) and 4 literals. Concretely:
\[
\begin{array}{l@{\hspace{3em}}l}
lit_{10} = (i_3 \leq 100) \\
lit_{14} = (i_3 > (i_4 - \dfrac{i_4}{10})) \\
lit_{21} = \bigwedge(i_3 \geq 0, i_3 \leq 200) \\
lit_{23} = \bigwedge(i_4 \geq 0, i_4 \leq 200)
\end{array}
\]

Note that, if we would like to enhance speed, $lit_{21}$ and $lit_{23}$ could be conjuncted in a single literal, since they are assumptions of the environment that will always hold together.

\subsubsection{Cluster 3: \textit{Lift (1,3)}.} This cluster contains 1 variable (which belongs to the environment) and 3 literals.
\[
\begin{array}{l@{\hspace{3em}}l}
lit_3 = (i_0 = 2) \\
lit_8 = (i_0 = 4) \\
lit_{15} = \bigwedge(i_0 \geq 0, i_0 \leq 4)
\end{array}
\]

Note that there are literals that use the equality operator, which is more restrictive than the comparison ones.

\subsubsection{Cluster 4: \textit{Lift (1,2)}.} This cluster contains 1 variable (which belongs to the environment) and 2 literals.
\[
\begin{array}{l@{\hspace{3em}}l}
lit_6 = (i_1 \neq 1) \\
lit_{17} = \bigwedge(i_1 \geq 0, i_1 \leq 4)
\end{array}
\]

Note the inequality operator, which is less restrictive than the comparison ones.

\subsection{Industrial case 2: Train}

Train is part of a set of specifications describing the functioning of an autonomous train driving system.
%
%Concretely, different sensors are read to know whether or not the doors can be closed after stopping at a station.
%
It is divided into 8 clusters, interpreted in $\mathcal{T}_{\mathbb{R}}$.

\subsubsection{Cluster 1: \textit{Train (1,3)}.} This cluster contains 1 variable (which belongs to the environment) and 3 literals. Concretely:
\[
\begin{array}{l@{\hspace{3em}}l}
    lit_1 = (in_6>1) \\
    lit_2 = (in_6 \leq 0.8 \cdot 1) \\
    lit_3 = \bigwedge(in_6 \geq 0.0, in_6 \leq 100)
\end{array}
\]

%Note that $lit_3$ is a constraint of the environment.

\subsubsection{Cluster 2: \textit{Train (2,1)}.} This cluster contains 2 variables (both of them belong to the system) and 1 literal. Concretely:
\[
\begin{array}{l@{\hspace{3em}}l}
    lit_4 = \bigwedge(v_{10}=100, v_6=100) 
\end{array}
\]

Note that $lit_4$ can be split into two clusters.

\subsubsection{Cluster 3: \textit{Train (1,3)}.} This cluster contains 1 variable (which belongs to the system) and 3 literals. Concretely:
\[
\begin{array}{l@{\hspace{3em}}l}
    lit_5 = (v_7=0) \\
    lit_6 = (v_7=100\cdot 1) \\
    lit_7 = (v_7=2.2)
\end{array}
\]

\subsubsection{Cluster 4: \textit{Train (1,1)}.} This cluster contains 1 variable (which belongs to the system) and 1 literal. Concretely:
\[
\begin{array}{l@{\hspace{3em}}l}
    lit_8 = (v_2=1)
\end{array}
\]

\subsubsection{Cluster 5: \textit{Train (3,6)}.} This cluster contains 3 variables (where $v_8$ belongs to the system, and $in_7$ and $in_8$ belong to the environment) and 6 literals. Concretely:
\[
\begin{array}{l@{\hspace{3em}}l}
    lit_{32} = (in_7 \neq 0) \\
    lit_{33} = (v_8=\dfrac{in_8}{in_7} \cdot 2) \\
    lit_{34} = (in_8 \neq 0) \\
    lit_{36} = ((in_8 - in_7) > 20) \\
    lit_{43} = \bigwedge(in_7 \geq 0.0, in_7 \leq 10) \\
    lit_{44} = \bigwedge(in_8 \geq 0.0, in_8 \leq 100)
\end{array}
\]

%Again, note that, if we would like to enhance speed, $lit_{43}$ and $lit_{44}$ could be conjuncted in a single literal, since they are assumptions of the environment. %that will always hold together.

Again, note that $lit_{43}$ and $lit_{44}$ can be conjuncted in a single literal.

\subsubsection{Cluster 6: \textit{Train (4,5)}.} This cluster contains 4 variables (where $v_4$ and $v_5$ belong to the system, and $in_{10}$ and $in_{11}$ belong to the environment) and 5 literals. Concretely:
\[
\begin{array}{l@{\hspace{3em}}l}
    lit_{39} = (v_4=in_{10}) \\
    lit_{40} = (v_4=1) \\
    lit_{41} = (v_5=in_{11}) \\
    lit_{42} = (v_5=1) \\
    lit_{45} = \bigwedge(\bigwedge(in_{10} \geq 0.0, in_{10} \leq 100), \bigwedge(in_{11} \geq 0.0, in_{11} \leq 100))
\end{array}
\]

Note, in $lit_{45}$, the arbitrarily large formulae about bounds of variables.

\subsubsection{Cluster 7: \textit{Train (3,5)}.} This cluster contains 3 variables (where $v_{12}$ belongs to the system, and $in_1$ and $in_{12}$ belong to the environment) and 5 literals. Concretely:
\[
\begin{array}{l@{\hspace{3em}}l}
    lit_{39} = (v_4=in_{10}) \\
    lit_{40} = (v_4=1) \\
    lit_{41} = (v_5=in_{11}) \\
    lit_{42} = (v_5=1) \\
    lit_{45} = \bigwedge(\bigwedge(in_{10} \geq 0.0, in_{10} \leq 100), \bigwedge(in_{11} \geq 0.0, in_{11} \leq 100))
\end{array}
\]

\subsubsection{Cluster 8: \textit{Train (3,12)}.} This cluster contains 4 variables (all of them belong to the system) and 5 literals. Concretely:
\[
\begin{array}{l@{\hspace{3em}}l}
    lit_{12} = (v_{14} \geq v_{12}) \\
    lit_{13} = (v_{13} = v_{12}) \\
    lit_{14} = (v_{12} > 0) \\
    lit_{15} = (v_{12} < 2.2 \cdot 1.5) \\
    lit_{17} = (v_{15} = 10) \\
    lit_{18} = (v_{13} = v_{15}) \\
    lit_{19} = (v_{13} = 0) \\
    lit_{20} = (v_{12} > v_{15}) \\
    lit_{22} = (v_{14} = v_{13}) \\
    lit_{23} = (v_{12} \geq 0) \\
    lit_{24} = (v_{14} = v_{12}) \\
    lit_{25} = ((v_{13} \cdot 1.2) > v_{12})
\end{array}
\]

\subsection{Industrial case 3: Connect}

Connect is part of a set of specifications describing the functioning of an electric vehicle charging and discharging system.
%
%Concretely, different readings are obtained on the voltages of the load points to decide how to level the input and output power flow.
%
It contains a single cluster \textit{Connect (2,2)}, interpreted in $\mathcal{T}_{\mathbb{Z}}$.

The cluster contains 2 variables (both of them belong to the environment) and 2 literals. Concretely:
\[
\begin{array}{l@{\hspace{3em}}l}
lit_1 = (a \leq 100) \\
lit_2 = (a > (\textit{b} - \dfrac{\textit{b}}{10}))
\end{array}
\]

Note that we are performing an integer division.

%\subsection{Industrial case 4: Two tanks}

%This specification is part of a set of specifications describing the functioning of a system consisting of two independently operating liquid pumping tanks.

%Concretely, it fluctuates between different heights in tank capacities, based on sensor readings of liquid flow rates.

\subsection{Industrial case 4: Cooker}

Cooker is part of a set of specifications describing the operation of a food processor with various functions.

%Concretely, using internal information about the time and external information about the user, it varies between different cooking modes and intensities.
%
It contains a single cluster \textit{Cooker (3,5)}, interpreted in both $\mathcal{T}_{\mathbb{Z}}$ and $\mathcal{T}_{\mathbb{R}}$.

The cluster contains 3 variables (the three of them belong to the system) and 5 literals, where $b \in \mathbb{R}$. Concretely:
\[
\begin{array}{l@{\hspace{3em}}l}
lit_1 = (a = c+10) \\
lit_2 = (a = c-10) \\
lit_3 = (a = c+1) \\
lit_4 = (a = c-1) \\
lit_5 = (b < a)
\end{array}
\]

Note that, since we are making an (unsound) comparison between an integer value and a real value in $lit_5$, Z3 converts the integer typed value into a real typed one.

\subsection{Industrial case 5: Usb}

Usb is part of a set of specifications describing the operation of a system that prevents the loss of information during the interaction between a USB and a machine.
%
%Concretely, it features different security modes, as well as error handling and different possible connections between the USB and the machine.
%
It is divided into 2 clusters, interpreted in in $\mathcal{T}_{\mathbb{Z}}$.

\subsubsection{Cluster 1: \textit{Usb (2,3)}.} This cluster contains 2 variables (both of them belong to the system) and 3 literals. Concretely:
\[
\begin{array}{l@{\hspace{3em}}l}
    lit_1 = (a > 100) \\
    lit_2 = (b = a^2) \\
    lit_3 = (b = 0)
\end{array}
\]

Note that, since the \textit{power} operation is only accepted for reals, Z3 again performs a type conversion from integer to reals.

\subsubsection{Cluster 2: \textit{Usb (3,5)}.} This cluster contains 3 variables (the three of them belong to the system) and 5 literals. Concretely:
\[
\begin{array}{l@{\hspace{3em}}l}
    lit_1 = (a > 100) \\
    lit_2 = (b = a^2) \\
    lit_3 = (b = 0) \\
    lit_4 = (c = 0) \\
    lit_5 = (c = 1)
\end{array}
\]

Note that this case is a simple stressing from the previous one, adding an integer variable that could be (in this case) interpreted as a Boolean out of the Boolean abstractions, and also as a cluster itself.

\subsection{Industrial case 6: Stages}

Stages is part of a set of specifications describing the operation of a system that combines the use of different sensors for use in aviation.
%
%Concretely, it is shown how the readings of a pitot tube are managed, based on some resistor's temperatures, power voltages and tubes' altitude lectures.
%
It is divided into 2 clusters, interpreted in in $\mathcal{T}_{\mathbb{Z}}$.

\subsubsection{Cluster 1: \textit{Stage (8,8)}.} This cluster contains 8 variables (all of them belong to the environment) and 8 literals. Concretely:
\[
\begin{array}{l@{\hspace{3em}}l}
    lit_1 = (a > (b + 100)) \\
    lit_2 = (a \leq 200) \\
    lit_3 = (a > (c - \dfrac{\textit{c}}{10})) \\
    lit_4 = (\textit{d} \leq 200) \\
    lit_5 = (\textit{e} > 1) \\
    lit_6 = (\textit{v} > 1) \\
    lit_7 = (\textit{w} > 1) \\
    lit_8 = (\textit{z} > 1)
\end{array}
\]

Note that there are several clusters merged in this one ($lit_1$ to $lit_3$, and the rest are a single cluster each predicate). Since only the environment player appears in them, the Boolean abstraction remains fast.

\subsubsection{Cluster 2: \textit{Stage (3,6)}.} This cluster contains 3 variables (which belong to the environment) and 6 literals. Concretely:
\[
\begin{array}{l@{\hspace{3em}}l}
    lit_9 = (a \geq 100) \\
    lit_10 = (a > 25) \\
    lit_11 = (a=(a - b)) \\
    lit_12 = (a = 25) \\
    lit_13 = (a = a + c) \\
    lit_14 = (a < 50)
\end{array}
\]

\subsection{Synthetic examples}

This specification is different from the rest of them. Here, we stress an original specification in order to test the Boolean abstraction tool. Note that we can interpret literals in both $\mathcal{T}_{\mathbb{Z}}$ and $\mathcal{T}_{\mathbb{R}}$.
The original specification \textit{Syn (2,2)} contains 2 variables (where x belongs to the environment and y belongs to the system) and 2 literals. Concretely:
\[
\begin{array}{l@{\hspace{3em}}l}
l_1 = (y > -2) \\
l_2 = (y < x)
\end{array}
\]

Then, we add a new constraint to make \textit{Syn (2,3)}:
\[
\begin{array}{l@{\hspace{3em}}l}
l_1 = (y > -2) \\
l_2 = (y < x) \\
l_3 = (0 < x)
\end{array}
\]

We add another one to make \textit{Syn (2,4)}:
\[
\begin{array}{l@{\hspace{3em}}l}
l_1 = (y > -2) \\
l_1 = (y > -2) \\
l_2 = (y < x) \\
l_3 = (0 < x) \\
l_4 = (x < 10)
\end{array}
\]

We add another one to make \textit{Syn (2,5)}:
\[
\begin{array}{l@{\hspace{3em}}l}
l_1 = (y > -2) \\
l_2 = (y < x) \\
l_3 = (0 < x) \\
l_4 = (x < 10) \\
l_5 = (x > 5)
\end{array}
\]

We add another one to make \textit{Syn (2,6)}:
\[
\begin{array}{l@{\hspace{3em}}l}
l_1 = (y > -2) \\
l_2 = (y < x) \\
l_3 = (0 < x) \\
l_4 = (x < 10) \\
l_5 = (x > 5) \\
l_6 = (x < 20)
\end{array}
\]

And we add the last one to make \textit{Syn (2,7)}:
\[
\begin{array}{l@{\hspace{3em}}l}
l_1 = (y > -2) \\
l_2 = (y < x) \\
l_3 = (0 < x) \\
l_4 = (x < 10) \\
l_5 = (x > 5) \\
l_6 = (x < 20) \\
l_7 = (x > 15)
\end{array}
\]

Note that \textit{Syn (2,7)} has two executions in Fig.~\ref{tab:benchmark} in order to illustrate results with different heuristics. 
Also, note that semantics of original \textit{Syn (2,2)} may vary with each addition of a constraint, yet they add no new theory-level operators.% in order to maintain Z3’s effort low.
This may affect realizability results.

\section{Correctness} \label{appSec:proofs}

The main element of the proof of correctness of our Boolean
abstraction technique is to show that every strategy of system in the
game that corresponds to $\phiT$ can be mimicked by the system in the
game of $\phiB$ in a way that one is winning if and only if the other
is winning.
This essentially boils down to proving that a local move for the
system can be mimicked in both games.

%Note that these proofs were suggested at \cite{rodriguezSanchez2022reactiveSynthesisLTLRichTheories}. 
%Also, note that full formalization is reported as a work in progress, but suggested at \cite{rodriguez2021temporalBooleanization}.

%\newcommand{\vs}{\ensuremath{\overline{v}}\xspace}
%\newcommand{\ws}{\ensuremath{\overline{w}}\xspace}

\subsection{Local Simulation}

We start by stating properties of the set of valid reactions.

\begin{lemma}
  \label{lemmReactionExistence}
  For every valuation of $\vs$ of $\xs$ there is at least one reaction $C$
  such that $\react_C[\xs \shortleftarrow \vs]$ is
  valid.
  Therefore,
  $\varphi_{React} = \forall \xs .  \bigvee_{r \in React } r$
  is valid.
\end{lemma}

\begin{proof}
  Let $\vs$ be an arbitrary valuation of the variables $\xs$
  and let
  $C=\{ c \in \calC | I[\xs \shortleftarrow \vs]
  \vDash \exists \ys. c(\xs,\ys)) \}$ where $I$ is an arbitrary
  interpretation.
  It follows that
  $I[\xs\shortleftarrow\vs] \vDash react_C$, since
  for every $c \in C$ then
  $I[\xs \shortleftarrow \vs] \vDash c$ and for any
  $c \notin C$ then
  $I[\xs \shortleftarrow \vs] \nvDash c$.
\end{proof}

The following lemma shows that there is an always valid move of the
system (the extra requirement is never blocking).
For every movement of the environment, the system can move at least
with one of the reactions.

\begin{theorem} \label{theoForEveryMovement}
  $\varphi_{\textit{VR}} = \forall \xs .  \bigvee_{r \in \textit{VR} }
  r$ is a valid formula.
\end{theorem}

\begin{proof}
By contradiction, assume $\varphi_{\textit{VR}} = \forall \xs .  \bigvee_{r \in \textit{VR} } r$ is not valid. 
Then, there is an interpretation $I$ such that
$I \not\vDash \varphi_{\textit{VR}}$, or equivalently
$I[\xs \shortleftarrow \vs] \nvDash \bigwedge_{r \in \textit{VR}}r$,
for some $\vs$.
By Lemma \ref{lemmReactionExistence}, we know $\varphi_{React}$ is
valid, so
$I[\xs \shortleftarrow \vs] \vDash \bigwedge_{r \in \textit{React}}r$,
for some $\vs$.
That is, there is a reaction $r$ such that $r[\xs \shortleftarrow v]$
is valid. Therefore, $\exists \xs. r$ is valid, so $r\in \textit{VR}$. This
means: $I[\xs \shortleftarrow v] \vDash r $.

It follows that there is a $r \in React \setminus \textit{VR}$ such that $I[\vs \shortleftarrow \xs] \vDash r$, which implies that $I \vDash \exists \xs .r$.
Since $I \vDash \exists \xs .r$ is closed,
$I \vDash \exists \xs .r$ is valid.
This is a contradiction.
%
%\qed
\end{proof}

As an observation, in the extra requirement, the set of potentials in
valid reactions cannot be empty.
In other words, for every move of the environment the system can
always move with a valid reaction, which will result in the
always-existence of some outcome.
This is stated in Lemma \ref{lemmNotEmptyReactions}.

\begin{lemma} \label{lemmNotEmptyReactions} Let $C \in \mathcal{C}$ be
  such that $react_C \in \textit{VR}$. Then $C \neq \emptyset$.
\end{lemma}

\begin{proof}
  Bear in mind $\react_C \in \textit{VR}$ is valid. 
  Let $\vs$ be such that $\react_C[\xs \shortleftarrow \vs]$ is
  valid.
  Let $\ws$ be an arbitrary valuation of $\ys$ and let $c$ be a
  configuration and $l$ a literal.
  Therefore:
  \[\bigwedge_{l[\xs \shortleftarrow \vs, \ys \shortleftarrow \ws]\textit{ is true }}l \wedge \bigwedge_{l[\xs \shortleftarrow \vs, \ys \shortleftarrow \ws]\textit{ is false }}\neg l\]
  It follows that $I[\xs \leftarrow \vs] \exists \ys . c$, so $c\in C$.
  %
  %\qed
\end{proof}

Lemma~\ref{lemmNotEmptyReactions} is crucial, because it ensures that
once a Boolean abstraction algorithm is executed, for each fresh $\overline{e}$ variable
in the extra requirement, at least one reaction with one or more
potentials can be responded by the system.

\subsection{From Local Simulation to Equi-Realizability} 

Realizability from LTL specifications considers infinite games.
The positions in the arena of the game are valuations of the atomic
propositions.
The two players take turns choosing alternatively the values of their
variables, resulting in a new position.
Then, an infinite play is winning for the system if the specification
is satisfied in the induced trace of the played, when the
specification formula is evaluated according to the semantics of the
logic.
A strategy of the system is a map that assigns a move, given the
previous sets of positions and the current move of the environment.
A strategy is winning if all plays played according to it are winning
for the system, in which case the specification is realizable and a
system can be extracted.
Note that, for $\phiB$ in particular, a winning system strategy always
moves to positions where the extra requirement is true (otherwise the
$\phiB$ would not hold and the strategy would not be winning.

For $\LTLt$ realizability, the arena has infinitely many position,
since valuations of the variables are now considered.

It is easy to see that the results in the previous sub-section allow
to define a simulation between the positions of the games for $\phiT$
and $\phiB$ such that (1) each literal $l_i$ and the corresponding
variable $s_i$ have the same truth value in related positions, (2) the
extra requirement is always satisfied, and (3) moves of the system in
each game from related positions in each game can be mimicked in the other
game.

\begin{theorem} 
  \label{temporalTheorem} 
   System wins $\GameT$ if and only if System wins the game $\GameB$.
  Therefore, $\phiT$ is realizable if and only if $\phiB$ is
  realizable.
\end{theorem}

\begin{proof}
Since realizability games are memory-less determined, it is sufficient
to consider only local strategies.
Given a  strategy $\rhoB$ that is winning in $\GameB$ we define
a strategy $\rhoT$ in $\GameT$ as follows.
Assuming related positions, $\rhoT$ moves in $\GameT$ to the successor
that is related to the position where $\rhoB$ moves in $\GameB$.
By (3) above, it follows that for every play played in $\GameB$
according to $\rhoB$ there is a play in $\GameT$ played according to
$\rhoT$ that results in the same trace, and vice-versa: for every play
played in $\GameT$ according to $\rhoT$ there is a play in $\GameB$
played according to $\rhoB$ that results in the same trace.
Since $\rhoB$ is winning, so is $\rhoT$.

The other direction follows similarly, because again $\rhoB$ can be
constructed from $\rhoT$ not only guaranteeing the same valuation of
literals and corresponding variables, but also that the extra
requirement holds in the resulting position.
\qed
\end{proof}

%%% Local Variables:
%%% TeX-master: "main.tex"
%%% TeX-PDF-mode: t
%%% End:

\end{document}